\documentclass[journal]{IEEEtran}
\usepackage{cite}
\usepackage{algorithmic}
\usepackage{graphicx}
\usepackage{textcomp}
\usepackage{xcolor}
\usepackage{comment}
\usepackage{amssymb}

\usepackage{amsmath, amsthm}

\usepackage{siunitx}
\usepackage{soul}
\usepackage{multirow}
\usepackage{subcaption}
\usepackage{siunitx}
\usepackage[font=footnotesize]{caption}
\usepackage{hyperref}
\usepackage{mathtools}
\usepackage[ruled,vlined,linesnumbered]{algorithm2e}
\usepackage{bm}
\usepackage[normalem]{ulem}

\usepackage{nomencl}
\usepackage{etoolbox}

\newtheorem{thm}{Theorem}[section]

\newtheorem{prop}[thm]{Proposition}
\newtheorem{cor}[thm]{Corollary}

\newtheorem{defn}{Definition}[section]

\newtheorem{rem}{Remark}

\def\BibTeX{{\rm B\kern-.05em{\sc i\kern-.025em b}\kern-.08em
    T\kern-.1667em\lower.7ex\hbox{E}\kern-.125emX}}
    
\begin{document}

\title{Integrated Task and Motion Planning for Safe Legged Navigation in Partially Observable Environments}

\author{
Abdulaziz~Shamsah,~\IEEEmembership{Student Member,~IEEE,}  Zhaoyuan~Gu,~\IEEEmembership{Student Member,~IEEE,} 
Jonas~Warnke,~\IEEEmembership{Student Member,~IEEE,}
Seth~Hutchinson,~\IEEEmembership{Fellow,~IEEE,}
and~Ye~Zhao,~\IEEEmembership{Senior Member,~IEEE}
\thanks{This work was partially funded by the NSF grants \# IIS-1924978, \# CMMI-2144309, ONR grant \#  N00014-23-1-2223, Georgia Tech Research Institute IRAD grant, and Georgia Tech Institute for Robotics and Intelligent Machines (IRIM) Seed Grant.}
\thanks{The authors are with the Laboratory for Intelligent Decision and Autonomous Robots, Woodruff School of Mechanical Engineering, Georgia Institute of Technology, Atlanta, GA 30313, USA
        {\tt\small \{ashamsah3, zgu78, jwarnke, seth, yezhao\}@gatech.edu}}%
\thanks{Mailing address: 801 Ferst Dr. NW
Atlanta, Georgia, United State, 30332.
}
\thanks{Corresponding author: Y. Zhao}
}

\IEEEaftertitletext{\vspace{-2\baselineskip}} 

\maketitle

\begin{abstract}

This study proposes a hierarchically integrated framework for safe task and motion planning (TAMP) of bipedal locomotion in a partially observable environment with dynamic obstacles and uneven terrain. The high-level task planner employs linear temporal logic (LTL) for a reactive game synthesis between the robot and its environment and provides a formal guarantee on navigation safety and task completion. To address environmental partial observability, a belief abstraction is employed at the high-level navigation planner to estimate the dynamic obstacles’ location. 
Accordingly, a synthesized action planner sends a set of locomotion actions 
to the middle-level motion planner, while incorporating safe locomotion specifications extracted from safety theorems based on a reduced-order model (ROM) of the locomotion process. The motion planner employs the ROM to design safety criteria and a sampling algorithm to generate non-periodic motion plans that accurately track high-level actions. At the low level, a foot placement controller based on an angular-momentum linear inverted pendulum model is implemented and integrated with an ankle-actuated passivity-based controller for full-body trajectory tracking. To address external perturbations, this study also investigates safe sequential composition of the keyframe locomotion state and achieves robust transitions against external perturbations through reachability analysis. The overall TAMP framework is validated with extensive simulations and hardware experiments on bipedal walking robots Cassie and Digit  designed by Agility Robotics.

\end{abstract}



\section{Introduction}

\IEEEPARstart{R}{obots} are increasingly being deployed in real-world environments, with legged robots presenting superior 
versatility in complex workspaces. However, safe legged navigation in real-life workspaces still poses a challenge,  particularly in a partially observable environment comprised of dynamic and possibly adversarial obstacles as seen in Fig.~\ref{fig:sim}. While motion planning for bipedal systems in dynamic environments has been widely studied \cite{bohorquez2016safe, pajon2019safe, motahar2016composing}, the proposed solutions often lack formal guarantees on simultaneous locomotion and navigation safety, with the exception of a recent work in \cite{scianca2021behavior}.
Formal guarantees on safety and task completion in a complex environment have been gaining interest in recent years \cite{srinivasan2020control, kress2018synthesis, vasile2017minimum, vasilopoulos2020reactive}, however hierarchical planning frameworks with multi-level safety guarantees for underactuated legged robots remain lacking. An intrinsic challenge of such multi-level formal guarantees is how to guarantee viable execution of high-level commands for low-level, full-body control that involves inherently complex bipedal dynamics.


\begin{figure}[t]
\centerline{\includegraphics[width=.46\textwidth]{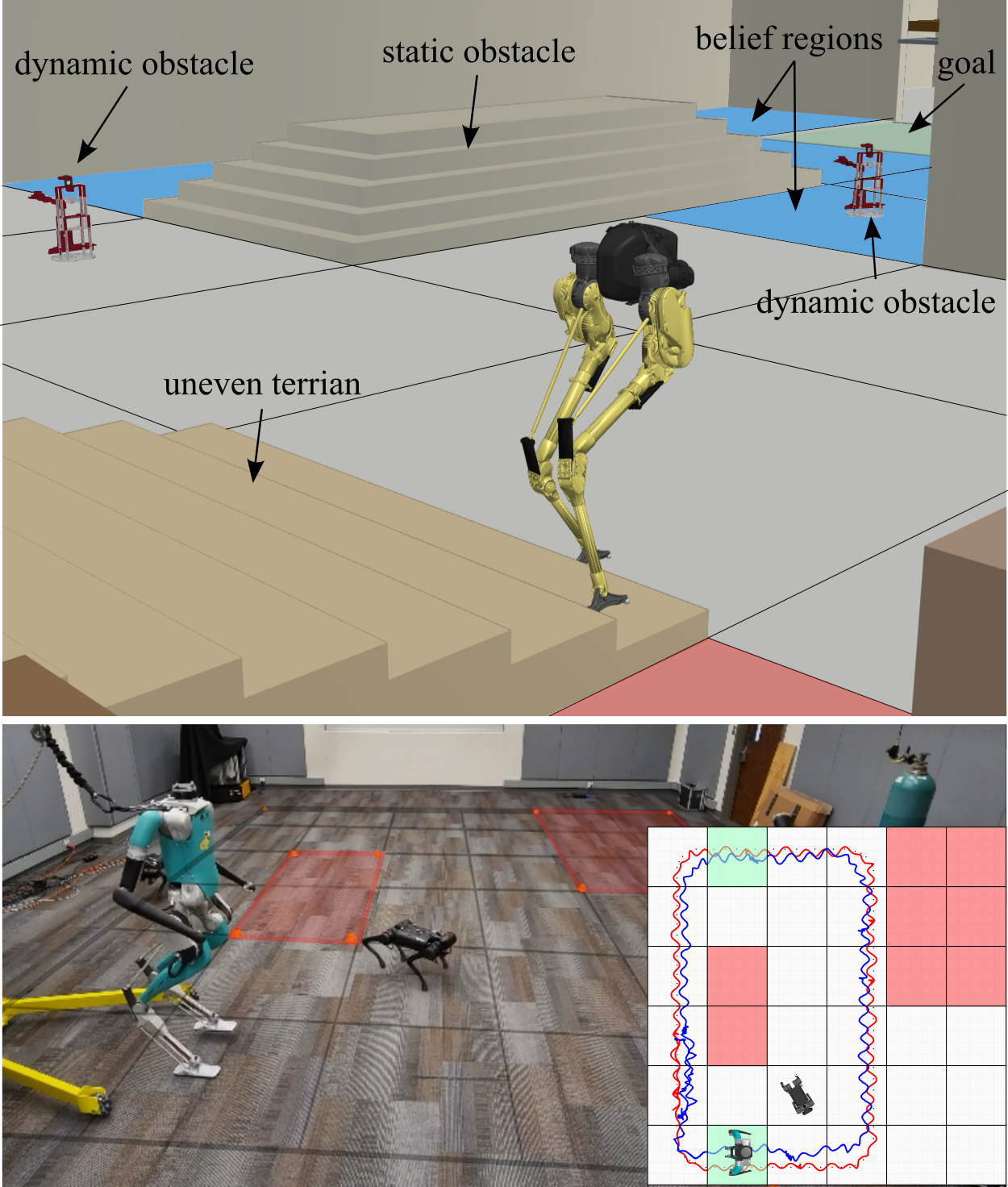}}
\caption{A snapshot of the simulation environment (top subfigure) and real-world experiment environment (bottom subfigure) for the proposed TAMP framework. The walking robot is deployed to accomplish safe navigation tasks. The environment contains static and dynamic obstacles, and uneven terrains.}
\label{fig:sim}
\vspace{-0.15in}
\end{figure}

\begin{figure*}[t]
\centerline{\includegraphics[width=.9\textwidth]{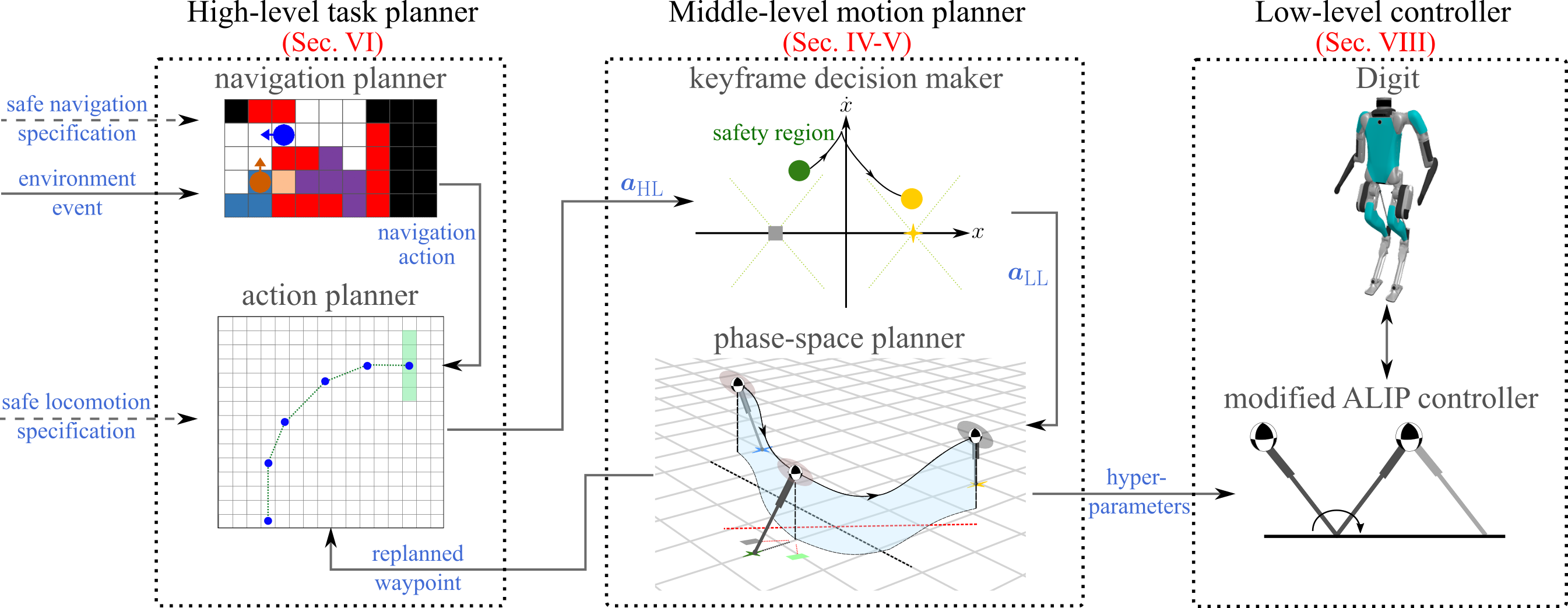}}
\caption{Block diagram of the proposed TAMP framework. The high-level task planner employs a linear temporal logic approach to synthesize locomotion actions for navigation tasks. The middle-level motion planner generates a safe motion plan based on a ROM. Safe motion plan's hyperparameters are sent to the modified ALIP controller to generate a stepping location that tracks the desired motion plan, based on the measured robot state. The desired foot placement is then achieved through a passivity-based controller on Digit. The high-level task planner and the middle-level motion planner are integrated in an \textit{online} fashion as shown by the solid black arrows. The dashed arrows represent \textit{offline} computations.}
\label{fig:frame}
\vspace{-0.15in}
\end{figure*}

This study proposes a hierarchically integrated task and motion planning (TAMP) framework as shown in Fig.~\ref{fig:frame} and provides multi-level formal safety guarantees on dynamic locomotion and navigation in dynamic and partially observable environments as shown in Fig.~\ref{fig:sim}. Guaranteeing safe navigation of legged robots in the presence of possibly adversarial obstacles becomes particularly challenging in partially observable environments. The range of a robot's sensor and occlusion caused by static obstacles give the adversarial agent a strategic advantage when trying to falsify the bipedal robot's safety guarantees by moving through non-visible regions of the environment. Our work is motivated by the surveillance game literature \cite{bharadwaj2018synthesis} to track possible non-visible dynamic obstacle locations via belief space planning. Belief space planning allows us to model the set of possible obstacle locations and track how this set evolves, guaranteeing collision avoidance in a larger set of environments.

Our framework takes safety into account in each layer of the hierarchical structure, and in how these layers are interconnected to achieve simultaneous safe locomotion and navigation. The high-level linear-temporal-logic (LTL)-based task planner incorporates reduced-order-model (ROM)-based dynamics constraints into the safety specifications, thus guaranteeing safe execution of the high-level actions in the underlying motion planners. The middle-level motion planner also employs a ROM-based planner called phase-space planning (PSP). This PSP can be seamlessly integrated with the abstracted high-level symbolic planner due to its hybrid planning nature. We integrate the task and the motion planners for \textit{online execution}, even in the presence of external perturbations, such as center-of-mass (CoM) velocity jumps caused by external forces. At the low level, we build a foot placement controller based on the angular-momentum linear inverted pendulum model (ALIP)~\cite{Gong2022AngularMomentum}, with a few critical modifications for phase-space plans. Geometric-based inverse kinematics functions design full-body reference trajectories using hyperparameters from the middle-level phase-space plan. An ankle-actuated passivity-based controller~\cite{sadeghian2017passivity} tracks the full-body reference trajectory. We are able to regulate the full-body motion to emulate the ROM and minimize tracking errors in foot placements and CoM velocities on a bipedal robot Digit \cite{agility}.

Robustness at the motion planning layer is of key importance, as continuous perturbations (e.g., CoM perturbations) can be naturally handled at this layer \cite{heim2019beyond}, unlike in the discretized high-level task planner which is unaware of locomotion dynamics. To this end, we formulate the locomotion gait in the lens of controllable regions \cite{zaytsev2018boundaries} and sequential composition \cite{burridge1999sequential} where we sequentially compose controllable regions to robustly complete a walking step. We employ ROM-based backward reachability analysis to compute robust controllable regions and synthesize appropriate controllers to safely reach the targeted state.

The main contributions of this study are as follows:
\begin{itemize}


\item Design a hierarchically integrated planning framework that provides formal safety guarantees simultaneously for the high-level task planner and middle-level ROM-based motion planner, which enables safe locomotion and navigation involving steering walking.

\item Design safe sequential composition of controllable regions for robust locomotion in the presence of perturbations, and sampling-based keyframe decision maker for accurate waypoint tracking to facilitate middle-level navigation safety.
\item Synthesize an LTL-based reactive navigation game for safe legged navigation and employ a belief abstraction method to expand navigation decisions in partially observable environments.
\item Experimental evaluation of the proposed framework on a bipedal robot Digit to navigate safely in a complex environment with dynamic obstacles.
\end{itemize}

A conference version of the work presented in this paper was published in \cite{warnke2020towards}. The work presented here extends the middle-level motion planner's safety and robustness against external CoM perturbations through formulating our previously introduced keyframe PSP scheme in the lens of controllable regions, safe sequential composition and reachability analysis. We also introduce a sampling-based keyframe decision maker to replace the heuristic-based keyframe decision maker in the conference version for accurate high-level waypoint tracking. From the high-level task planner, we present non-deterministic LTL transitions to facilitate online replanning capability of the high-level waypoint, as well as joint belief abstractions for efficient estimation of multiple dynamic obstacles' locations.  Finally, we validate the feasibility of ROM-based locomotion models on a bipedal robot Digit\footnote{In this paper, we use Cassie and Digit interchangeably given their similar leg kinematics and dynamics, and we do not consider the upper body dynamics of the Digit robot.} \cite{agility} with 28 degrees of freedom.

This paper is outlined as follows. Sec.~\ref{sec:related_work} is a literature review of related work.
Sec.~\ref{sec:prelim} introduces the ROM-based locomotion planning and keyframe definitions. Then safety theorems for locomotion planning and reachability-based analysis for robustness against perturbations are introduced in Sec.~\ref{sec:motion planner}. In Sec.~\ref{sec:Tracking}, we introduce our sampling-based keyframe decision maker algorithm. 
Sec.~\ref{sec:task planner} outlines our LTL-based high-level task planner which guarantees safe navigation in a partially observable environment. In Sec.~\ref{sec:recoverbility}, we evaluate the performance of the proposed recoverability strategy. Low-level controllers and hardware implementation details are in Sec.~\ref{sec:hardware}. The results of our integrated framework are shown in Sec.~\ref{sec:results}. We discuss the limitations in Sec.~\ref{sec:discussion} and conclude in Sec.~\ref{sec:conclusion}. 

\section{Related Work}
\label{sec:related_work}

Motion planning in complex environments has been extensively studied, with a spectrum of approaches in the literature \cite{chinchali2019multi, teng2021toward, alwala2020joint, Kulgod2020LTL, cao2022leveraging}. Reactive methods for motion planning with formal guarantees are widely studied with the methods of safety barrier certificates \cite{certificate_Wang2017}, artificial potential functions \cite{rimon1990exact}, and more recent extensions \cite{vasilopoulos2020reactive, huang2021efficient}. While the work in \cite{vasilopoulos2020reactive} provides convergence guarantees and obstacle avoidance 
in geometrically complicated unknown environments, 
it is restricted to static obstacles and has only been demonstrated on a fully actuated particle, or a simple unicycle model for a quadruped. Whereas the framework we present here is able to generate safe locomotion plans reacting to environmental events that include multiple, possibly adversarial, dynamic obstacles with formal guarantees on task completion and safety. Moreover, our framework is validated on a bipedal robot Digit \cite{agility}.

Numerous bipedal motion planning works are based on pendulum models. Classic approaches---such as zero moment point \cite{Kajita2003}, divergent component of motion \cite{DCM}, capture point \cite{pratt2006capture, koolen2012capturability, zaytsev2018boundaries}---have been extensively studied. The work in \cite{zaytsev2018boundaries} is closely related to our work, as we use controllable regions and viability theory \cite{aubin2011viability} to guarantee safety and to achieve non-periodic walking gaits. A majority of existing works have focused on locomotion safety \cite{bohorquez2016safe, pajon2019safe, motahar2016composing, huang2021efficient, li2021vision, kuindersma2016optimization}, but high-level task planning has been largely ignored. The work in \cite{huang2021efficient} introduces a variation on the non-holonomic differential-drive wheeled robot model to include the capabilities and limitations of bipedal robots. This study demonstrates remarkable safe navigation and locomotion in a variety of static environments with uneven terrains; however, the navigation safety relies solely on the ability of the periodic-gait controller \cite{Gong2022AngularMomentum} to track the velocity commands outputted from the proposed omnidirectional control Lyapunov function (CLF) based on a wheeled robot model. Similarly, the work in \cite{li2021vision} relies on a gait library to generate foot placements which are constrained based on the robot dynamics and kinematics limits. The authors in \cite{li2021vision} ensure navigation safety by generating a path that maintains a safe distance from static obstacles.

In our work, we attempt to differentiate between navigation and locomotion safety. The former focuses on designing robot navigation paths to avoid obstacle collisions, and the latter relates to motion design during individual walking steps to maintain balance. We claim that, specifically in regards to locomotion, the foot placement needs to be designed properly to achieve both navigation and locomotion safety, while solely relying on a ROM-based planner to select foot placements can only achieve locomotion safety.

Model predictive control (MPC), as a well-studied online method for locomotion motion planning, uses models with different complexities---such as linear inverted pendulum model \cite{MPC_LIPM}, single rigid body model \cite{Ding_MPC_SRBM}, and centroidal model \cite{Centroidal_MPC}---to promptly update motions for a certain time horizon. Recent works use MPC as a foot placement planner for terrain adaption~\cite{gibson2022terrain} or obstacle avoidance through control barrier functions (CBFs)~\cite{teng2021toward, narkhede2022sequential}. These works, in a sense, separate navigation and balancing safety, similar to the principles we target in this paper. However, the ROM-based MPC methods~\cite{gibson2022terrain,teng2021toward, narkhede2022sequential} lack the integration of a high-level task planner. Our high-level planner runs in an MPC fashion; it interacts with the environment as a one-step-horizon MPC and provides safety guarantees for both navigation and locomotion tasks. 

\begin{figure*}[t]
\centerline{\includegraphics[width=.8\textwidth]{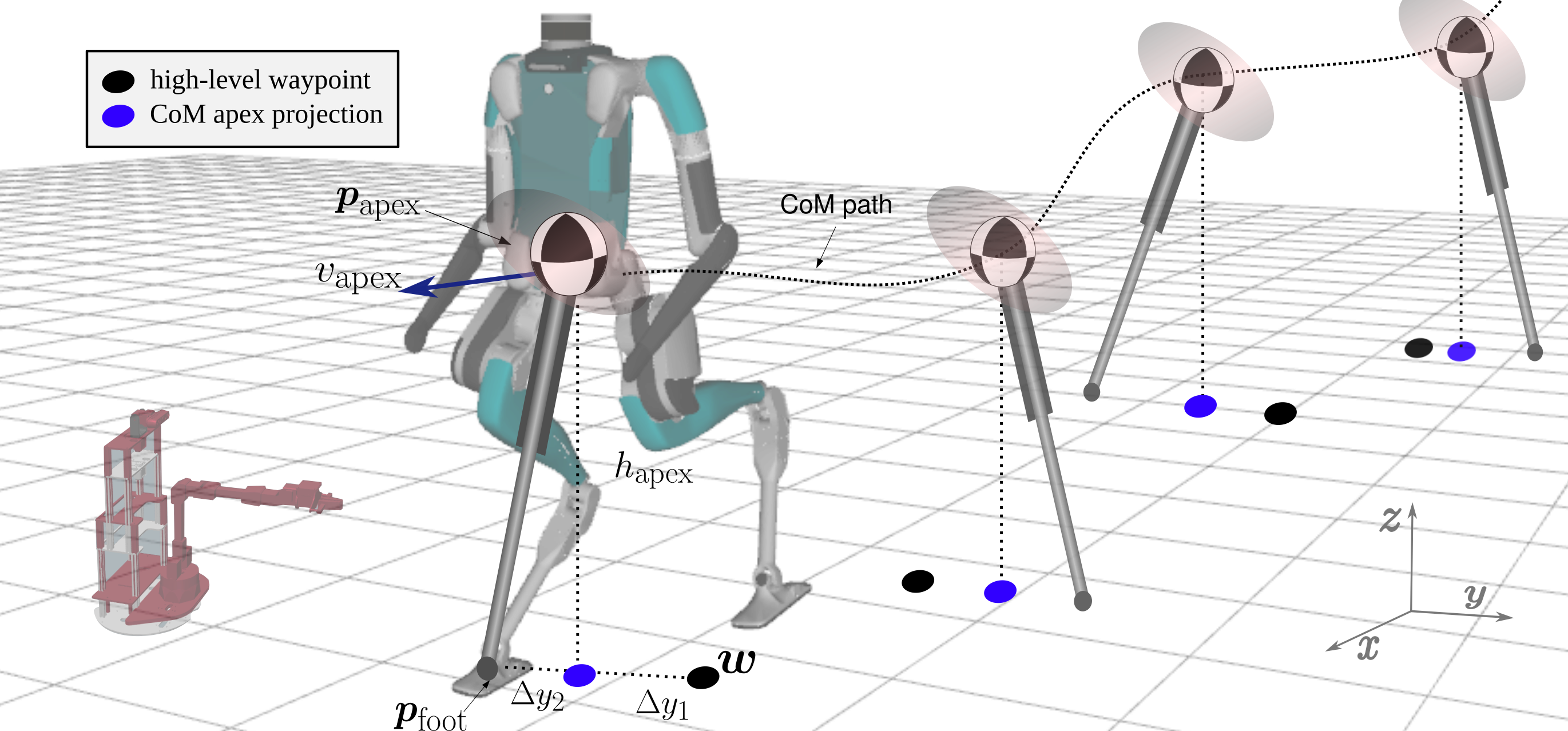}}
\caption{Reduced-order modeling of our Digit robot as a 3D prismatic inverted pendulum model with all of its mass concentrated on its CoM and a telescopic leg to comply to the varying CoM height. The CoM motion follows a parameterized CoM path depending on keyframe states. $\Delta y_1$ is the relative lateral distance between lateral CoM apex position and the high-level waypoint $\boldsymbol{w}$, and $\Delta y_2$ is the lateral distance between the CoM lateral apex position and the lateral foot placement.}
\label{fig:notation}
\vspace{-0.15in}
\end{figure*}

Formal synthesis methods have been well established to guarantee high-level robot behaviors in dynamic environments \cite{kaelbling2013integrated, kress2009temporal, fainekos2005temporal}. Collision-free navigation in the presence of dynamic obstacles has been achieved via multiple approaches such as local collision avoidance controllers in \cite{decastro2018collision}, incrementally expanding a motion tree in sampling-based approaches \cite{plaku2016motion}, and Velocity Obstacle Sets generated by obstacle reachability analysis in \cite{wu2012guaranteed}. Collision avoidance and task completion become more challenging to formally guarantee when the environment is only partially observable as such an environment has a strategic advantage in being adversarial. Navigating through partially known maps with performance guarantees has been achieved through exploring \cite{Exploring_Partially_Known}, updating the discrete abstraction, and re-synthesizing a controller at runtime in \cite{Temporal_hybrid_systems}. To avoid the computational costs of online re-synthesis, others have proposed patching a modified local controller into an existing global controller when unmodeled non-reachable cells, i.e. static obstacles, are discovered at runtime \cite{livingston2012backtracking, livingston2013patching}. The authors in \cite{Temporal_hybrid_systems} have proposed a satisfaction metric to meet the specification as closely as possible when run-time discovered environment constraints render the specification unsatisfiable. 

Partial observability and hierarchical planning are addressed recently in \cite{rosolia2022unified}, where the authors demonstrate safe navigation and task planning using high-level LTL task planning, two MPC problems at the middle level, and CBF-CLF tracking controller at the low level. The work in \cite{rosolia2022unified} leverages mixed observable Markov decision processes to model the system-environment interaction in a partially-observable environment, i.e., some discrete regions in the environment are not known \textit{a priori} to be traversable. This approach above is better suited for guaranteeing successful navigation and collision avoidance in environments that are uncertain only with respect to static obstacles as they can not reason about when and where a dynamic obstacle may appear. In our work, the environment is partially observable when static obstacles occlude the robot's view of certain regions in the environment, thus guaranteeing collision avoidance with dynamic obstacles becomes a challenging problem. 

Collision avoidance with dynamic obstacles in partially observable environments has been achieved through approaches such as partially observable Markov decision processes (POMDPs) \cite{ragi2013uav}, probabilistic velocity obstacle modeling \cite{fulgenzi2007dynamic}, and object occlusion cost metrics \cite{chung2009safe}. The authors in \cite{bouraine2012provably} guarantee passive motion safety by avoiding braking Inevitable Collision States (ICS) at all times via a braking ICS-checking algorithm. While these solutions provide collision avoidance guarantees, they assume dynamic obstacles could appear at any time and result in an overly conservative strategy. Our method investigates belief-space planning to provide the controller with additional information on when and where dynamic obstacles may appear in the robot's visible range to inform the synthesized strategy if navigation actions are guaranteed to be safe, even when static obstacles occlude the robot's view of adjacent environment locations. We have devised a variant of the approach in \cite{bharadwaj2018synthesis} to explicitly track a belief of which non-visible environment locations are obstacle free, reducing the conservativeness of a guaranteed collision-free strategy. This belief tracking method is then integrated into our hierarchical TAMP framework.

For whole-body joint trajectory design and low-level control, many approaches have been put forward in recent years: Bayesian optimization~\cite{yang2022bayesian}, sum-of-squares optimization~\cite{smit2019walking}, multi-layered CBFs and MPC~\cite{grandia2021multi}, CBF with CLF \cite{CBF_CLF_Ames}, data-driven step planner~\cite{DataS2S_Dai}, and reinforcement learning~\cite{li2021reinforcement, siekmann2021blind} to name a few. Recently, Grizzle's group at Michigan proposed an angular-momentum linear inverted pendulum model (ALIP)~\cite{Gong2022AngularMomentum}, which uses the angular momentum around the contact point in the robot's state. This state incorporates both linear momentum and angular momentum about the CoM, forming a more comprehensive representation that is less sensitive to internal joint-level noise and external contact impact. Hence, controlling foot placements with the ALIP model yields better velocity tracking accuracy. Our approach in this paper leverages this ALIP model with a few critical modifications to achieve high-performance, non-periodic locomotion control on our bipedal robot Digit hardware.

\section{Preliminaries}
\label{sec:prelim}
This section will introduce a phase-space planning approach \cite{zhao2012three, zhao2017robust} for CoM trajectory generation based on a ROM. Our framework is based on inverted pendulum models except for the low-level controller. Starting with a derivation of the dynamics of a Prismatic Inverted Pendulum Model (PIPM), and then define the locomotion keyframe state (a discretized feature state of our PSP approach) used as a connection between the high-level planner and the middle-level motion planner. Consequently, we define keyframe-based transitions to achieve safe locomotion. This section builds the basis for the safe locomotion planning proposed in later sections.

\subsection{Reduced-order Locomotion Planning}
\label{subsec:reduced-model}

This subsection introduces a mathematical formulation of our ROM. As shown in Fig. \ref{fig:notation}, the CoM position $\boldsymbol{p}_{\rm com} = (x, y, z)^T$ is composed of the sagittal, lateral, and vertical positions. We denote the apex CoM position as $\boldsymbol{p}_{{\rm apex}}=(x_{{\rm apex}},y_{{\rm apex}}, z_{{\rm apex}})^T$, the foot placement as $\boldsymbol{p}_{{\rm foot}}=(x_{{\rm foot}},y_{{\rm foot}}, z_{{\rm foot}})^T$, and $h_{\rm apex}$ is the relative apex CoM height with respect to the stance foot height. $v_{{\rm apex}}$ denotes the CoM velocity at $\boldsymbol{p}_{{\rm apex}}$. $\Delta y_1$ is the lateral distance between CoM and the high-level waypoint $\boldsymbol{w}$\footnote{The high-level discrete representation of the robot location.} at apex. $\Delta y_2 :=y_{{\rm apex}}-y_{{\rm foot}}$ denotes the lateral CoM-to-foot distance at apex. This parameter will be used to determine the allowable steering angle in Sec.~\ref{subsubsec:safetyprop}.

PIPM has been proposed for agile, non-periodic locomotion over rough terrain~\cite{zhao2017robust}. Here we reiterate for completeness the derivation of the centroidal momentum dynamics of this model. The single contact case using the moment balance equation along with linear force equilibrium is expressed as
\begin{equation}
    (\boldsymbol{p}_{\rm com} - \boldsymbol{p}_{\rm foot}) \times (\boldsymbol{f}_{\rm com} + m\boldsymbol{g})=-\boldsymbol{\tau}_{\rm com}
    \label{eqn:moment_balance}
\end{equation}
where $\boldsymbol{\tau}_{\rm com}$ is the angular moments of the torso exerted on the CoM, and $\boldsymbol{g}$ is the gravitational vector. For nominal planning we set $\boldsymbol{\tau}_{\rm com}=0$. Formulating the dynamics in Eq. (\ref{eqn:moment_balance})  for $j^{\rm th}$ walking step as a hybrid control system
\begin{equation}
    \Ddot{\boldsymbol{p}}_{{\rm com},j} =\Phi(\boldsymbol{p}_{{\rm com},j},\boldsymbol{u}_j)= \begin{pmatrix}
\omega^2_{j}(x-x_{{\rm foot},j})\\
\omega^2_{j}(y-y_{{\rm foot},j})\\
a\omega^2_{j}(x-x_{{\rm foot},j})
\end{pmatrix}
\label{eq:centrodial_dynamics}
\end{equation}
where the asymptote slope $\omega_j = \sqrt{g/h_{{\rm apex}, j}}$. The hybrid control input is $\boldsymbol{u}_j=(\omega_{j},\boldsymbol{p}_{{\rm foot},j})$, with $\boldsymbol{p}_{{\rm foot},j}$ being the discrete input\footnote{Hereafter, we will ignore the subscript q for notation simplicity. We will instead use $\cdot_c$ and $\cdot_n$ denoting the current and next apex
, respectively.}. The CoM motion is constrained within a piece-wise linear surface parameterized by $h = a(x -x_{{\rm foot}})+h_{\rm apex}$, where $h$ denotes the CoM height from the stance foot height, the ROM becomes linear and an analytical solution exists. Detailed derivations are elaborated in Appendix~\ref{appendix1}.

\noindent\textbf{Summary of Phase-space Planning:}
In PSP, the sagittal planning takes precedence over the lateral planning. The decisions for the planning algorithm are primarily made in the sagittal phase-space, such as step length and CoM apex velocity, where we propagate the dynamics forward from the current apex state and backward from the next apex state until the two phase-space trajectories intersect. The intersection state defines the foot stance switching instant. On the other hand, the lateral phase-space parameters are searched for to adhere to the sagittal phase-space plan and have consistent timings between the sagittal and lateral plans. In this paper, we build on our previous PSP work~\cite{zhao2017robust,warnke2020towards} to derive safety criteria for sagittal planning in order to achieve successful transitions between keyframe states in the presence of perturbations in Sec.~\ref{sec:motion planner}. Moreover, we employ a sampling algorithm based on the lateral apex states to select the next sagittal apex velocity that allows the lateral dynamics to comply to high-level waypoint tracking in Sec.~\ref{sec:Tracking}.

\subsection{Locomotion Keyframe for 3D Navigation}

PSP uses keyframe states for  non-periodic dynamic locomotion planning\cite{zhao2017robust}. Our study generalizes the keyframe definition in our previous work by introducing diverse navigation actions in 3D environments.

\begin{defn}[Locomotion keyframe state for 3D environment navigation]\label{def:keyframe}
A keyframe state of our ROM is defined as $\boldsymbol{k} = (d, \Delta\theta, \Delta z_{\rm foot}$,$ 
v_{\rm apex}, z_{\rm apex}$) $ \in \mathcal{K}$, where 
\begin{itemize}
    \item $d := x_{{\rm apex},n}-x_{{\rm apex},c}$ is the walking step length\footnote{In straight walking $d$ represents the step length. However, during steering walking $d$ is adjusted to reach the next waypoint on the new local coordinate.};
    \item $\Delta \theta:=\theta_{{\rm apex},n}-\theta_{{\rm apex},c}$ is the heading angle change at two consecutive CoM apex states;
    \item $\Delta z_{\rm foot} := z_{{\rm foot},n}-z_{{\rm foot},c}$ is the height change for successive foot placements;
    \item $v_{\rm apex}$ is the CoM sagittal apex velocity;
    \item $z_{\rm apex}$ is the global CoM height at apex.
\end{itemize}

\end{defn}
The keyframe state above can be divided into two action sets: a high-level (HL) action ($\boldsymbol{a}_{\rm HL}$) and a low-level (LL) action ($\boldsymbol{a}_{\rm LL}$). The HL action includes $\boldsymbol{a}_{\rm HL} = (d, \Delta\theta, \Delta z_{\rm foot}
)\in \mathcal{A}_{\rm HL}$, which is determined by the navigation policy to be designed in the task planner. The parameters $d$, $\Delta\theta$, and $\Delta z_{\rm foot}$ are expressed in the Cartesian space as the high-level waypoints $\boldsymbol{w}$. On the other hand, the LL action is $\boldsymbol{a}_{\rm LL} = (v_{\rm apex}, z_{\rm apex}) \in \mathcal{A}_{\rm LL}$, which is determined in the middle-level motion planner. The keyframe parameters are sent from the high-level task planner to the middle-level motion planner \textit{online} as shown in Fig.~\ref{fig:frame}.

\subsection{Keyframe Transition}
We now aim to formulate locomotion transition definitions in terms of the locomotion keyframe state $\boldsymbol{k}$ and describe the connection between the discretized keyframe state and the continuous dynamics of our reduced-order system introduced in Sec.~\ref{subsec:reduced-model}. We will first define locomotion safety.
\begin{figure}[t]
\centerline{\includegraphics[width=.4\textwidth]{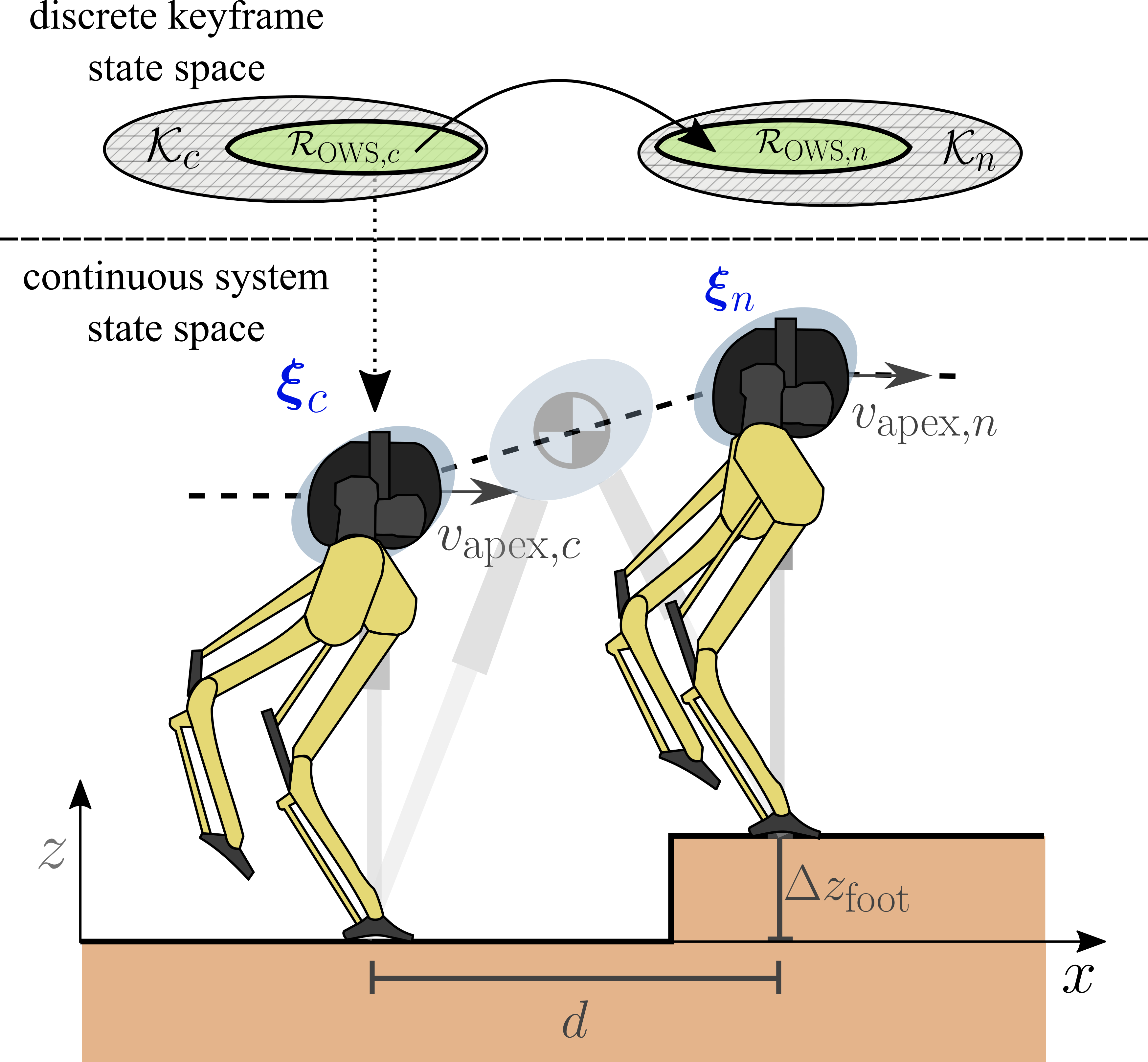}}
\caption{Sagittal system state transition for One Walking Step (OWS) without heading angle change. OWS is shown as the transition between two consecutive apex states, $\boldsymbol{\xi}_{c}$ and $\boldsymbol{\xi}_{n}$. The system state transition in Def.~\ref{def:transition_map} is shown as the projection of keyframe state onto the system state $\boldsymbol{\xi}_{c}$, and $d$ and $\Delta z_{\rm foot}$ are used to select the next foot placement $\boldsymbol{p}_{{\rm foot},n}$ in the hybrid control input $\boldsymbol{u}$. In the discrete keyframe state space, we show the transition between two consecutive keyframe states, where $\mathcal{R}_{{\rm OWS},c}$ is the set of current viable keyframe states that allows a successful system state transition to the next viable keyframe state set $\mathcal{R}_{{\rm OWS},n}$.}
\label{fig:system_state_transition}
\vspace{-0.15in}
\end{figure}
\begin{defn}[Locomotion safety]
Safety for a locomotion process is defined as a formal guarantee that the robot maintains balance dynamically, i.e., the CoM state in phase-space staying within the desired quadrant\footnote{The safe regions are shown in Fig.~\ref{fig:steering_safety} and further explained in Sec.~\ref{subsubsec:safetyprop}.}, while transitioning between consecutive locomotion keyframe states $\boldsymbol{k} \in \mathcal{K}$.
\label{def:balance_saf}
\end{defn}
Note that, the keyframe state $\boldsymbol{k}$ includes high-level actions $\boldsymbol{a}_{\rm HL}$ so the control is implicit in the \textit{Locomotion Safety}. Based on our keyframe definition in Def.~\ref{def:keyframe}, we define One Walking Step (OWS) as the transition between two consecutive keyframe states as shown in Fig.~\ref{fig:system_state_transition}. Therefore, we define the set of viable keyframe states for OWS as follows.
\begin{defn}[Viable keyframe set for one walking step]
$\mathcal{R}_{\rm OWS}$ is the set of keyframe states $\mathcal{K}$ that results in a viable transition to the next desired keyframe state through the continuous PIPM dynamics in Eq.~(\ref{eq:centrodial_dynamics}), thus achieving locomotion safety for OWS.
\label{def:Rows}
\end{defn}

The transition between keyframes is hybrid since it includes a continuous progression of the system states under the PIPM dynamics in Eq.~(\ref{eq:centrodial_dynamics}), followed by a discrete foot contact switch. 
In this study, we aim to provide formal guarantees that the selected keyframe states are within $\mathcal{R}_{\rm OWS}$. Since quantifying $\mathcal{R}_{\rm OWS}$ is computationally intractable due to its high dimensionality, we propose a set of safety theorems to quantify the viable region when $\mathcal{R}_{\rm OWS}$ is projected onto a reduced dimensional parameter space, which is selected as 
\begin{equation*}
   \textrm{Sagittal CoM system state: } \boldsymbol{\xi}=(x,\dot{x}) \in \boldsymbol{\Xi}
\end{equation*}

The current discrete keyframe state $\boldsymbol{k}_c \in \mathcal{K}$ corresponds to (i) the continuous system state at apex ($\boldsymbol{\xi}_c$) and (ii) the PSP hybrid control input $\boldsymbol{u}$ at the CoM apex in both straight and steering walking scenarios,
where the apex state in the keyframe Def.~\ref{def:keyframe} is the system state at the CoM apex\footnote{$\boldsymbol{\xi}_c = \boldsymbol{\xi}_{{\rm apex},c}$ only when $\Delta \theta = 0$, otherwise $\boldsymbol{\xi}_c$ is a non-apex state as can be seen in Fig.~\ref{fig:steering_safety}(a).}, and the step length and step height are used to calculate $\boldsymbol{p}_{\rm foot}$ in the hybrid control input $\boldsymbol{u}$. An illustration of these variables is shown in Fig.~\ref{fig:system_state_transition}. The desired next system state is always selected to be an apex state $\boldsymbol{\xi}_n = (d, v_{{\rm apex},n})$, where $ d \in \boldsymbol{a}_{\rm HL}$ is determined by the high-level planner and $v_{{\rm apex},n}$ is determined by the keyframe decision maker as detailed in Sec.~\ref{sec:Tracking}. Therefore, we can define a system state transition.

\begin{defn}[System state transition $\boldsymbol{\xi}_n = Tr(\boldsymbol{k}_c)$]
$Tr$ is a system state transition that takes the projection of the current keyframe state onto the continuous system state at apex and hybrid control input $(\boldsymbol{\xi}_c.\boldsymbol{u})$, to the desired next system state $\boldsymbol{\xi}_n$ through PSP of the centroidal dynamics $\Phi$ in Eq.~(\ref{eq:centrodial_dynamics}).
\label{def:transition_map}
\end{defn}
In Fig.~\ref{fig:system_state_transition}, we illustrate the system state transition for OWS. By projecting $\boldsymbol{k}_c$ state onto the $\boldsymbol{\xi}_c$, the centroidal dynamics in Eq.~(\ref{eq:centrodial_dynamics}) allows the system state to reach $\boldsymbol{\xi}_n$ based on $\boldsymbol{u}$.

\begin{figure}[t]
\centerline{\includegraphics[width=.3\textwidth]{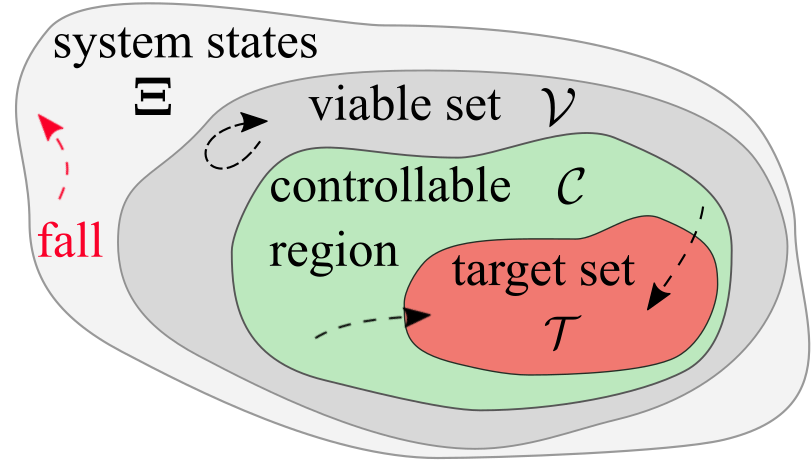}}
\caption{Viable sets and controllable regions for the set of system states $\boldsymbol{\Xi}$. Any state within the Viable set $\mathcal{V}$ (dark gray region), is guaranteed to remain inside $\mathcal{V}$ in finite time thus avoiding a fall. While any state within the controllable region set $\mathcal{C}$ (green region) is guaranteed to reach the target set $\mathcal{T}$ (red region) in finite time given an appropriate control input. The red trajectory indicates an initial state that results in a fall. In general, the viable set $\mathcal{V}$ is computationally difficult to estimate (e.g., for the states composing an OWS but not reaching $\mathcal{T}$), and thus this study focuses on computing the controllable region $\mathcal{C}$ for safe sequential composition.}
\label{fig:abs_viability}
\vspace{-0.05in}
\end{figure}

\begin{figure*}[th]
\centerline{\includegraphics[width=0.97\textwidth]{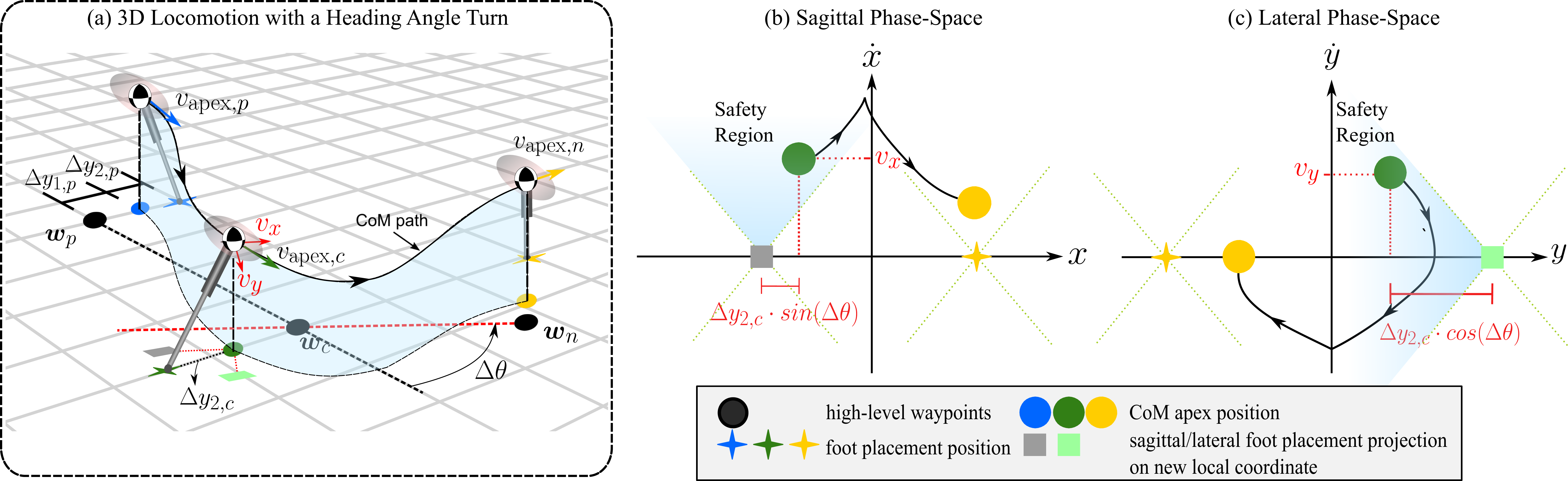}}
\caption{Phase-space safety region for steering walking: (a) shows three consecutive keyframes with a heading angle change ($\Delta \theta$) between the current keyframe and the next keyframe. The CoM trajectory and its projection on the sagittal-lateral plane are represented by the blue surface. The direction change introduces a new local coordinate, where the dashed black line is the sagittal coordinate before the turn, and the dashed red line is the sagittal coordinate after the turn. Subfigures (b) and (c) show the sagittal and lateral phase-space plots respectively, both satisfying the safety criteria proposed in Theorem \ref{thm:steering1}. The CoM apex state in the original coordinate becomes non-apex in the new coordinate (due to the coordinate change). The subscripts $p$, $c$, and $n$ denote the previous, current, and next walking steps, respectively.}
\label{fig:steering_safety}
\vspace{-0.15in}
\end{figure*}

\section{Safe Locomotion Planning}
\label{sec:motion planner}
This section will propose a set of safety theorems based on PSP for the ROM that allows us to select a safe next keyframe state under nominal conditions. Then in  Sec.~\ref{subsubsec:capturebasin} we define and compute controllable regions under bounded state disturbance by reachability analysis for a set of keyframe transitions that satisfy the safety theorems introduced in Sec.~\ref{subsubsec:safetyprop}. 
Within this controllable region, any state is guaranteed to reach a target set ($\mathcal{T}$) in finite time given a feasible control sequence as shown in Fig.~\ref{fig:abs_viability}. 
Accordingly, we sequentially compose consecutive controllable regions through our hybrid control input to guarantee locomotion safety and correctness. 
In addition to the safety criteria defined only for sagittal locomotion, safety for lateral tracking of the high-level waypoint is equally important for 3D locomotion. Therefore, Sec.~\ref{sec:Tracking} presents a sampling-based algorithm that adjusts the sagittal phase-space plan to allow for lateral waypoint tracking.

\subsection{Locomotion Safety Criteria}
\label{subsubsec:safetyprop}

In this subsection, we propose safe locomotion criteria based on the PIPM introduced in Sec.~\ref{subsec:reduced-model}, and
provide safety constraints for the locomotion keyframe state. 

As a general principle of balancing safety, the sagittal CoM position should be able to cross the sagittal apex with a positive CoM velocity while the lateral CoM velocity should be able to reach zero lateral velocity at the next apex. Ruling out the fall situations provides us the bounds of balancing safety regions. First, we study the constraints between the apex velocities of two consecutive walking steps and propose the following theorems and corollaries. 
\begin{thm}\label{thm:straight}
For safety-guaranteed straight walking, given $d$ and $\omega$, the apex velocity for two consecutive walking steps ought to satisfy the following velocity constraint:
\begin{equation}\label{eq:straight}
    -\omega^{2}d^{2} \le \underbrace{v_{{\rm apex},n}^{2}-v_{{\rm apex},c}^{2}}_{\begin{gathered}
       \vspace{-0.1in} \text{\scriptsize apex velocity square difference} \\ \vspace{-0.1in}\text{\scriptsize for two consecutive steps}
    \end{gathered}} \le \omega^{2}d^{2}
\end{equation}
where $d^2 = (x_{{\rm apex},n}-x_{{\rm apex},c}) (x_{{\rm apex},c}+x_{{\rm apex},n}-2x_{{\rm foot},c})$. Notably, $d$ is equal to the step length in Def.~\ref{def:keyframe}, i.e., $d = x_{{\rm apex},n}-x_{{\rm apex},c}$, during a straight walking where $x_{{\rm apex},c}=x_{{\rm foot},c}$. The proof of this criterion can be seen in Appendix~\ref{appendiX:proof_prop_straight}.
\end{thm}


Another consideration for safety is to limit the maximum allowable velocity of the CoM. Since the maximum velocity occurs at the foot switching instant, we explicitly enforce an upper velocity bound to this switching velocity $v_{\rm switch}$ to avoid over-accelerated motions, which can be further magnified by the ground impact dynamics in the real system. Through the analytical solution of the ROM in Appendix~\ref{appendix1}, we solve for $v_{\rm switch} = \Psi(v_{{\rm apex}, c}, v_{{\rm apex}, n}, d)$. 
Therefore, we set an upper bound on $v_{\rm switch}$, i.e., $v_{\rm switch} \leq v_{\rm max}$.

Similar to Theorem~\ref{thm:straight}, $v_{\rm switch}$ provides a nonlinear relationship between sagittal apex velocities for two consecutive apex states. Combining the boundary conditions in Theorem~\ref{thm:straight} and the limit of $v_{\rm switch}$ allows us to quantify the viable region of $v_{{\rm apex},n}$ given $v_{{\rm apex},c}$, $d$, and $\omega$.

The steering case requires a more restrictive criterion. A fall will occur when the turning angle $\Delta \theta$ is too large such that $v_{{\rm apex},c}$ in the new local coordinate after the turn is out of a safety range such that either the lateral CoM velocity cannot reach zero at the next apex or the sagittal CoM can not climb over the next apex.

\begin{thm}\label{thm:steering1}
For safety-guaranteed steering walking, the current sagittal CoM apex velocity $v_{{\rm apex},c}$ in the original local coordinate should be bounded by
\begin{equation}
    \Delta y_{2,c} \cdot \omega \cdot \tan{\Delta\theta} \leq v_{{\rm apex},c} \leq \frac{\Delta y_{2,c} \cdot \omega}{\tan{\Delta\theta}}
    \vspace{0.05in}
    \label{eq:steering_prop}
\end{equation}
\end{thm}
The proof of this theorem is shown in Appendix~\ref{appendiX:proof_prop_steering}. This theorem provides a bound on the heading angle change $\Delta \theta$ given the current apex velocity $v_{{\rm apex},c}$. Fig.~\ref{fig:steering_safety} shows a steering walking trajectory and phase-space plot that satisfy Theorem~\ref{thm:steering1}. Namely, the CoM location in the sagittal and lateral phase-space in the new local coordinate after the turn should not cross the asymptote line of the shaded safety region in Fig.~\ref{fig:steering_safety}. This criterion is specific to steering walking, as the heading change ($\Delta \theta$) introduces a new local frame and yields the current state $\boldsymbol{\xi}_c$ to no longer be an apex in the new coordinate. As such, it has non-apex sagittal and lateral components, i.e., $v_{y,c} \neq 0$, and $x_{\rm{apex},c} \neq x_{\rm{foot},c}$.\footnote{In this study, we use $\dot x$ and $v$ exchangeably to represent the CoM velocity.}

\begin{cor}
For steering walking in Theorem~\ref{thm:steering1}, given $d$, $\Delta \theta$, $\Delta y_{2,c}$ and $\omega$, two consecutive apex velocities ought to satisfy the following velocity constraint:
\begin{equation}\label{eq:steering1}
    -\omega^{2}d^{2} \le v_{{\rm apex},n}^{2}-(v_{{\rm apex},c}\cos{\Delta\theta})^{2} \le \omega^{2}d^2_{+}
\end{equation}
where $d^2_{+} = d^{2}+2\Delta y_{2,c}d\sin{\Delta\theta}$.
\end{cor}

\begin{cor}
For steering walking in Theorem~\ref{thm:steering1}, similarly, given $d$, $\Delta \theta$, $\Delta y_{2,c}$, and $\omega$, two consecutive apex velocities ought to satisfy the following velocity constraints,
\begin{equation}\label{eq:steering2}
    -\omega^{2}d^{2} \le v_{{\rm apex},n}^{2}-(v_{{\rm apex},c}\cos{\Delta\theta})^{2} \le \omega^{2}d^2_{-}
    \vspace{0.05in}
\end{equation}
where $d^2_{-} = d^{2} - 2\Delta y_{2,c}d\sin{\Delta\theta}$. Note that, parameters $v_{{\rm apex},n}$, $d$, and $\Delta \theta$ in Eqs.~(\ref{eq:straight})-(\ref{eq:steering2}) are the keyframe states.
\end{cor}

The aforementioned safety theorems provide quantifiable bounds on the next keyframe selection that leads to viable transitions under the governance of nominal disturbance-free PIPM dynamics. The next section will focus on definitions based on controllable regions and sequential composition to provide guarantees on the safe progression of the continuous system states $\boldsymbol{\xi}$ adhering to Theorems~\ref{thm:straight}-\ref{thm:steering1} under bounded disturbances $\boldsymbol{\tilde \Xi}$.

\subsection{Controllable Regions and Sequential Composition}
\label{subsubsec:capturebasin}

First, let us decompose OWS into two half steps at the instant when the hybrid control input 
$\boldsymbol{u}$ switches. Therefore, the First Half Walking Step (FHWS) starts from $\boldsymbol{\xi}_c$ until the foot contact switching state $\boldsymbol{\xi}_{\rm switch}$, the Second Half Walking Step (SHWS) starts with $\boldsymbol{\xi}_{\rm switch}$ until $\boldsymbol{\xi}_n$. Note that $t_{\rm FHWS}$ and $t_{\rm SHWS}$ are not always equal in non-periodic walking. We start by defining controllable regions of FHWS and SHWS.

\begin{defn}[Controllable region of FHWS]  Given $\boldsymbol{\tilde \Xi}_{\rm synthesis}$, $\mathcal{U}$ and $\mathcal{T}_{\rm switch}$, a controllable region of FHWS is defined as $\mathcal{C}_{\rm FHWS} \coloneqq \{\boldsymbol{\xi} | \boldsymbol{\dot \xi}(t) = \Phi(\boldsymbol{\xi}(t)+\boldsymbol{\tilde \xi}(t), \boldsymbol{u}(t)), \exists \boldsymbol{u}(t) \in \mathcal{U}, \textnormal{such that} \; \exists  \boldsymbol{\xi}(t_{\rm FHWS}) \in \mathcal{T}_{\rm switch}, \; t_{\rm FHWS}\;  \textnormal{is finite} \}$, where $\boldsymbol{\tilde \xi}(t) \in \boldsymbol{\tilde \Xi}_{\rm synthesis}$ represents a bounded state disturbance.
\label{def:CB_FHWS}
\end{defn}
\begin{defn}[Controllable region of SHWS] Given $\boldsymbol{\tilde \Xi}_{\rm synthesis}$, $\mathcal{U}$ and $\mathcal{T}_{\rm OWS}$, a controllable region of SHWS is defined as $\mathcal{C}_{\rm SHWS} \coloneqq \{\boldsymbol{\xi} | \boldsymbol{\dot \xi}(t) = \Phi(\boldsymbol{\xi}(t)+\boldsymbol{\tilde \xi}(t), \boldsymbol{u}(t)), \exists \boldsymbol{u}(t) \in \mathcal{U}, \textnormal{such that} \; \exists  \boldsymbol{\xi}(t_{\rm SHWS}) \in \mathcal{T}_{\rm OWS}, \; t_{\rm SHWS}\;  \textnormal{is finite} \}$, where $\boldsymbol{\tilde \xi}(t) \in \boldsymbol{\tilde \Xi}_{\rm synthesis}$ represents a bounded state disturbance.
\label{def:CB_SHWS}
\end{defn}
Numerical computation of the controllable region for OWS is achievable given $\boldsymbol{\Xi}_c$, $\mathcal{T}_{\rm OWS}$, $\mathcal{U}$, and a bounded disturbance $\boldsymbol{\tilde \Xi}$ through backward dynamic propagation using robustly complete control synthesis (ROCS) \cite{li2018rocs}. Unlike other reachability analysis tools~\cite{holmes2020reachable, bansal2017hamilton}, ROCS is a partition-based control synthesis tool for nonlinear systems. Over-approximation of the controllable region is computed through interval-valued functions of the PIPM dynamics. All reachable states after a predefined time step from any state in $\boldsymbol{\Xi}_c$ are captured in the output. For more details please refer to ~\cite{liu2017robust, jaulin:hal-00845131, zhao2018reactive}. Given the backward propagation nature of ROCS, $\mathcal{C}_{\rm SHWS}$ is computed first starting from $\mathcal{T}_{\rm OWS}$, which is selected to be the set of desired $\boldsymbol{\xi}_n$ at the next apex. Then we set $\mathcal{T}_{\rm switch}$ to be all the states of $\mathcal{C}_{\rm SHWS}$ that are within the tangent manifolds of FHWS, where the tangent manifolds represent the nominal phase-space trajectory given $v_{{\rm apex, min}}$, $v_{{\rm apex, max}}$ and $x_{{\rm foot},c}$ (see Fig.~\ref{fig:PSP_viability}) \cite{zhao2017robust,zhao2018reactive}. Then we compute $\mathcal{C}_{\rm FHWS}$ with $\mathcal{T}_{\rm switch}$ being the target set. Note that $\mathcal{T}_{\rm switch}$ represents the set of system states that a foot contact transition can occur at. Moreover, ROCS also allows us to synthesize a controller $\boldsymbol{u}(t) \in \mathcal{U}$ that guarantees that the system state $\boldsymbol{\xi}$ reaches the target set $\mathcal{T}_{\rm OWS}$ as long as the system state remains within the controllable region. The details of the controllable regions in Defs.~\ref{def:CB_FHWS}-\ref{def:CB_SHWS} are shown in Fig.~\ref{fig:PSP_viability}.

\begin{figure}[t]
\centerline{\includegraphics[width=.45\textwidth]{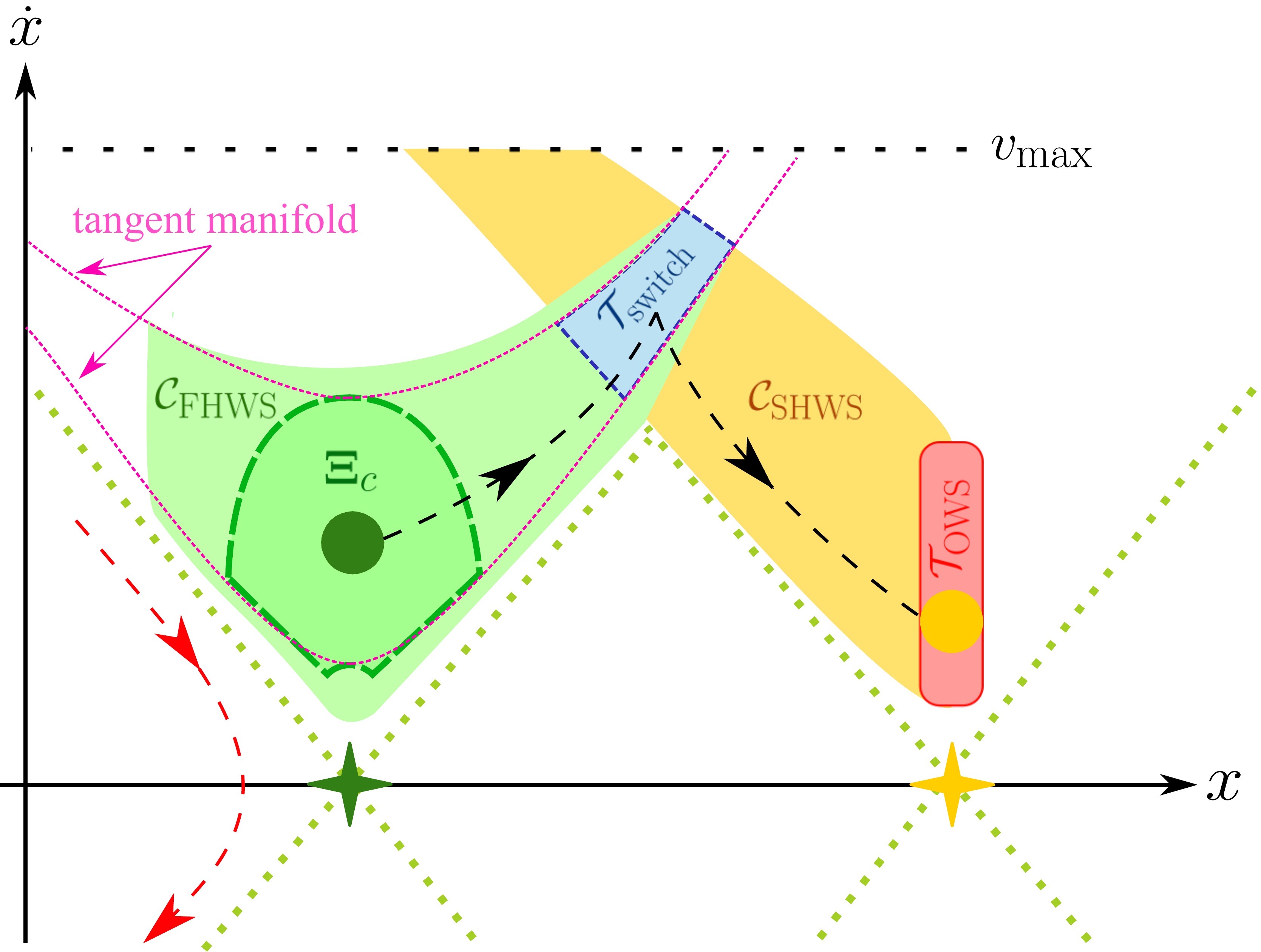}}
\caption{Projection of the controllable regions for OWS on the sagittal phase-space based on Defs.~\ref{def:CB_FHWS}-\ref{def:CB_SHWS}. Given $\boldsymbol{\xi}_c$ (green circle) $\in \boldsymbol{\Xi}_c$ (green dashed region) $\subset \mathcal{C}_{\rm FHWS}$ (green region), the continuous system state is guaranteed to reach $\mathcal{T}_{\rm switch}$ (blue region). $\mathcal{T}_{\rm switch}$ is bounded by pink tangent manifolds and $\mathcal{C}_{\rm SHWS}$. Switching from the current foot stance to the next foot stance (green and yellow stars respectively) at $\mathcal{T}_{\rm switch}$ guarantees that the continuous system state will reach $\boldsymbol{\xi}_n$ (yellow circle) $\in \mathcal{T}_{\rm OWS}$ (red region). The red dashed arrow shows a system state outside of the controllable region and results in a fall.}
\label{fig:PSP_viability}
\vspace{-0.15in}
\end{figure}

Sequentially composing the controllable regions defined in Defs.~\ref{def:CB_FHWS}-\ref{def:CB_SHWS}, i.e, $\mathcal{T}_{\rm switch} \neq \emptyset$, affords a guarantee on safe task completion for OWS. Controllable regions have a ``funnel"-type geometry that is guaranteed to reach a target set, and correct switching between such funnels ultimately leads to the target set  $\mathcal{T}_{\rm OWS}$ as seen in Fig.~\ref{fig:funnels}. A projection of the composition of the controllable regions on the sagittal phase-space satisfying Theorem~\ref{thm:sequential_Comp} can be seen in Fig.~\ref{fig:PSP_viability}.

The controllable regions in Def.~\ref{def:CB_FHWS}-\ref{def:CB_SHWS} depend on not only the state of the system but also the high-level keyframe state $\boldsymbol{k}$ as it determines the target set $\mathcal{T}$ as well as the foot placement in the hybrid control input $\boldsymbol{p}_{\rm foot}\in \boldsymbol{u}$. This further provides safety guarantees within our layered framework.

\begin{figure}[t]
\centerline{\includegraphics[width=.4\textwidth]{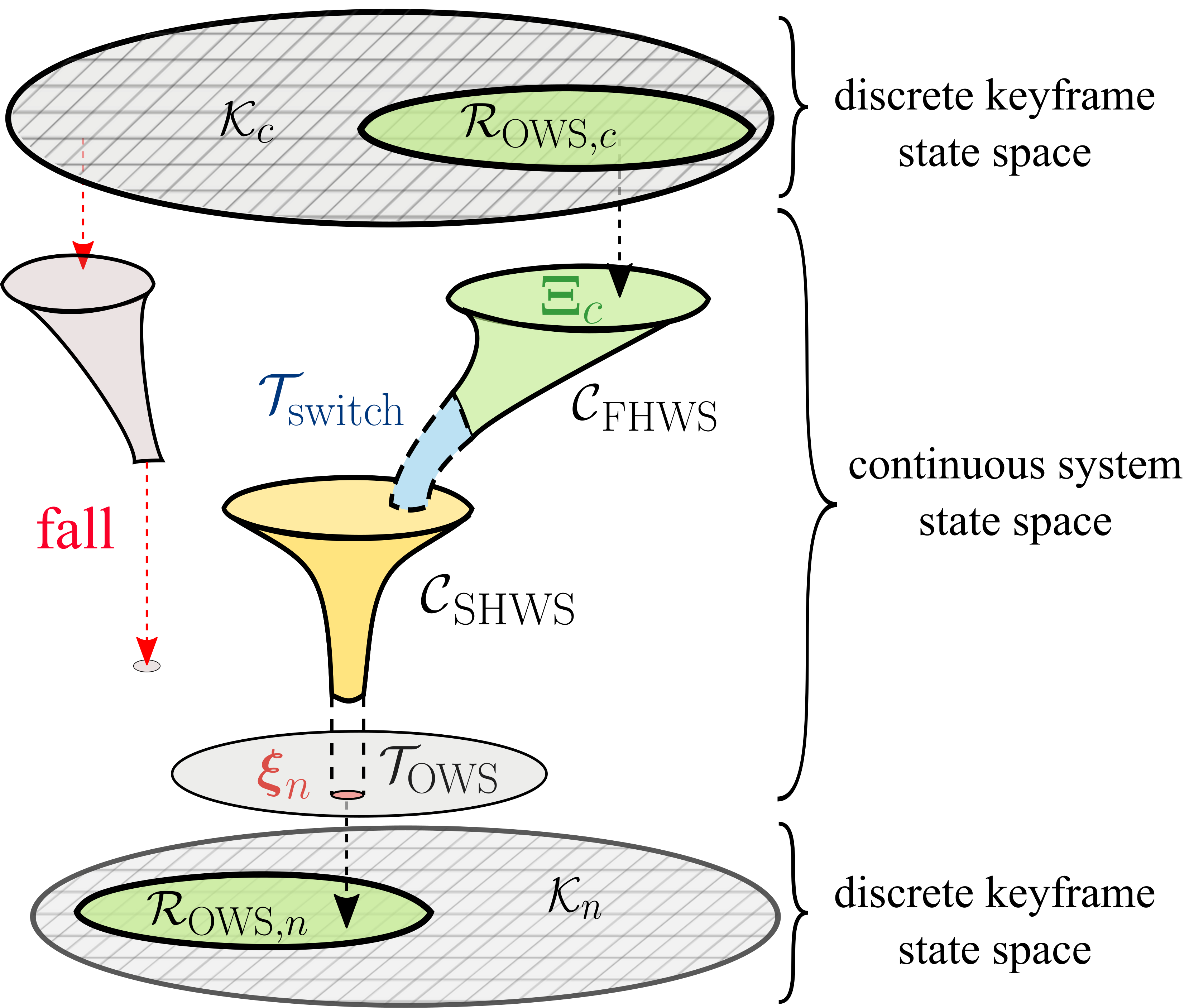}}
\caption{A conceptual illustration of the keyframe transition between two consecutive keyframes, where starting from the current viable keyframe state set $\mathcal{R}_{{\rm OWS},c}$, can be projected to the continuous state space and the dynamics are then modeled as the hybrid control system in Eq.~(\ref{eq:centrodial_dynamics}). The evolvement of the continuous state is formulated as a sequential composition of controllable regions for One Walking Step (OWS) as in Defs.~\ref{def:CB_FHWS}-\ref{def:CB_SHWS}. The state of the system starts from $\boldsymbol{\Xi}_c$ in $\mathcal{C}_{\rm FHWS}$ (green funnel). After a finite time, the state reaches $\mathcal{T}_{\rm switch}$ (dashed blue region), where the dynamics of the system switch from $\mathcal{C}_{\rm FHWS}$ to $\mathcal{C}_{\rm SHWS}$ (the yellow funnel). Finally, the state of the system reaches $\mathcal{T}_{\rm OWS}$ (red region). The switch between $\mathcal{C}_{\rm FHWS}$ and $\mathcal{C}_{\rm SHWS}$ can occur at any instant within $\mathcal{T}_{\rm switch}$ and the state would still be guaranteed to reach $\mathcal{T}_{\rm OWS}$. Any $\boldsymbol{k}_c \in \mathcal{R}_{\rm OWS}$ is guaranteed to reach $\mathcal{T}_{\rm OWS}$ based on  Def.~\ref{def:Rows}, while the states outside $\mathcal{R}_{\rm OWS}$ are considered failed state as shown in the red arrows.}
\label{fig:funnels}
\vspace{-0.15in}
\end{figure}

\begin{thm}
\label{thm:sequential_Comp}
The controllable regions for OWS are sequentially and safely composable, i.e., $\mathcal{C}_{\rm FHWS} \cap \mathcal{C}_{\rm SHWS} = \mathcal{T}_{\rm switch} \neq \emptyset$, if the locomotion safety defined in Def.~\ref{def:balance_saf} is satisfied, i.e., the PSP obeying (i) Theorem~\ref{thm:straight}, (ii) $v_{\rm switch} \leq v_{\rm max}$, and (iii) Theorem~\ref{thm:steering1} generates a feasible CoM trajectory.
\end{thm}


\begin{proof}
To guarantee the feasibility of a nominal phase-space plan generated through forward and backward propagation of the PIPM dynamics, the following two conditions on $v_{{\rm apex},n}$ need to be satisfied: (i) $\exists x_{\rm switch}$ such that $x_{{\rm apex},c} \le x_{\rm switch} \le x_{{\rm apex},n}$, which is guaranteed by \textit{designing} $v_{{\rm apex},n}$ that obeys Theorem~\ref{thm:straight} given a feasible $d$ and current keyframe state $\boldsymbol{\xi}_c \in \boldsymbol{\Xi}_c$; (ii) given a maximum CoM velocity threshold $v_{\rm max}$, $v_{{\rm apex},n}$ is \textit{chosen} such that $v_{\rm switch} \leq v_{\rm max}$ through the forward and backward propagation. 
The designed $v_{{\rm apex},n}$ meeting the two conditions above will guarantee feasible phase-space trajectories that safely compose the controllable regions of the two half-walking steps, i.e., $\mathcal{C}_{\rm FHWS} \cap \mathcal{C}_{\rm SHWS} = \mathcal{T}_{\rm switch} \neq \emptyset$. Even in the extreme case of that $v_{\rm switch}$ is highly close to $v_{\rm max}$, we still guarantee $\mathcal{T}_{\rm switch} \neq \emptyset$ since part of the $\mathcal{T}_{\rm switch}$ region is below the normal phase-space trajectory and guaranteed to be feasible. 

As for the turning case, given the designed $v_{{\rm apex},n}$, condition (iii) constrains the maximally allowable heading angle change $\Delta \theta$ for the next walking step. This condition guarantees that the CoM state in the sagittal and lateral phase-space in the new local coordinate after the turn will not cross the asymptote line of the shaded safety region (see Fig.~\ref{fig:steering_safety}). As such, the controllable region $\mathcal{C}_{\rm FHWS}$ centering around the nominal PSP trajectory will exist (in the space between the nominal PSP and the asymptote line) and interact with $\mathcal{C}_{\rm SHWS}$ such that $\mathcal{T}_{\rm switch} \neq \emptyset$.
\end{proof}


According to Theorem~\ref{thm:sequential_Comp}, we can determine a set of $v_{{\rm apex}, n}$, $\Delta \theta$, and $d$ parameters\footnote{Specific values of $d$, $v_{{\rm apex}}$, and $\Delta \theta$ are included in Table.~\ref{tab:value_table}, Sec.~\ref{sec:results}.} that satisfy conditions (i)-(iii) thus guaranteeing that the system state can reach the desired target set $\mathcal{T}_{\rm OWS}$. The selected set is included as safe locomotion specifications in the high-level planner as shown in Fig.~\ref{fig:frame} and detailed in Sec.~\ref{subsec:low_in_high}.  


\begin{cor}
To realize safe walking for an arbitrary number of steps, (i) the target set of the current step is required to be a subset of the controllable region of the first half walking step for the next step, i.e., $\mathcal{T}_{{\rm OWS},c} \subset \mathcal{C}_{{\rm FHWS},n}$, and (ii) the applied perturbation during execution does not push the system state outside of the controllable regions.
\label{col:composition}
\end{cor}
Fig.~\ref{fig:funnels} conceptually shows how after a system state transition $\boldsymbol{\xi}_n = Tr(\boldsymbol{k}_c)$, $\boldsymbol{\xi}_n$ can be projected onto the viable keyframe set for the next OWS $\mathcal{R}_{{\rm OWS}, n}$ to guarantee the viability of next walking step.

\section{Keyframe Decision-Making for Navigation Waypoint Tracking}
\label{sec:Tracking}

In the previous section, we proposed safety theorems that guarantee locomotion safety. Now we shift our focus to another consideration for safe locomotion by ensuring tracking of the high-level waypoints. The lateral phase-space plan is determined based on the sagittal phase-space plan, as the contact switch timing in the lateral dynamics needs to obey that of the sagittal dynamics. Therefore, the lateral dynamics depends on sagittal apex velocities and sagittal step length. In our previous work~\cite{zhao2017robust}, the lateral foot placement is solved through a Newton-Raphson search method, such that the lateral CoM velocity is equal to zero at the next CoM apex. While our previous method achieved stable walking and turning, it lacks the guarantee of accomplishing high-level navigation through tracking of the waypoints. Therefore, the lateral CoM motion may not track the desired waypoints. In \cite{warnke2020towards}, we propose a heuristic-based policy that restricts the allowable keyframe transitions to achieve waypoint tracking for specific locomotion plans. In this study, we extend our previous work by designing an algorithm that formally manipulates the sagittal phase-space plan to take into account high-level waypoint tracking. Particularly, we use $\Delta y_{1}$ and $\Delta y_{2}$ to track the lateral distance between the CoM at the apex and the high-level waypoint as seen in Fig.~\ref{fig:notation}. First, let's define viable ranges for $\Delta y_{1}$ and $\Delta y_{2}$. 
\begin{defn}[Viable range for lateral-apex-CoM-to-waypoint distance $\Delta y_1$] $\mathcal{R}_{\Delta y_1} \coloneqq \{\Delta y_{1} | \Delta y_{1}+\Delta y_{2} \leq b_{\rm safety} \}$, where $b_{\rm safety}$ denotes the safety boundary around the waypoint.
\label{def:Rdelta_y1}
\end{defn}
\begin{defn}[Viable range for lateral-apex-CoM-to-foot distance $\Delta y_2$] Given the safety criterion for steering walking defined in Theorem~\ref{thm:steering1}, the viable range for lateral CoM-to-foot distance at apex is defined as $\mathcal{R}_{\Delta y_2} \coloneqq \{\Delta y_{2} | v_{\rm apex, max} \cdot \tan{\Delta\theta}/\omega \leq  \Delta y_{2} \leq (v_{\rm apex, min})/(\omega \cdot \tan{\Delta\theta}) \}$.
\label{def:Rdelta_y2}
\end{defn}
$\mathcal{R}_{\Delta y_1}$ and $\mathcal{R}_{\Delta y_2}$ is defined as such 
to avoid the lateral drift of the robot's CoM and foot location from the high-level waypoint, and further avoid collisions with obstacles. Given Defs.~\ref{def:Rdelta_y1}-\ref{def:Rdelta_y2}, we can track the high-level waypoint as follows.
\begin{prop}
Viable lateral tracking of the high-level waypoint is guaranteed only if (i) $\Delta y_{2}$ and $\Delta y_{1}$ are bounded within their respective viable ranges, i.e., $\Delta y_{1} \in \mathcal{R}_{\Delta y_1}$ and $\Delta y_{2} \in \mathcal{R}_{\Delta y_2}$, and (ii) the sign of $(\Delta y_{1}+\Delta y_{2})$ alternates between two consecutive keyframes.
\label{thm:tracking_viability}
\end{prop}
Proposition~\ref{thm:tracking_viability} requires that (i) the distance sign of the lateral foot stance position relative to the waypoint alternates between consecutive keyframes and (ii) the waypoints and CoM trajectory are bounded within the lateral foot placement width. An example of this trajectory is shown in Fig.~\ref{fig:top_down_turn}.

The analytical solutions of $\Delta y_{1,n}$ and $\Delta y_{2,n}$ are highly nonlinear functions of multiple parameters including the step length $d$, heading angle change $\Delta \theta$, current and next apex velocities $v_{{\rm apex},c}$, $v_{{\rm apex},n}$ and the current lateral state of the system $\Delta y_{1,c}$ and $\Delta y_{2,c}$. Thus, it is difficult to quantitatively analyze the relationship between $\Delta y_1$, $\Delta y_2$, and other parameters aforementioned. Since $(d, \Delta \theta) \in \boldsymbol{a}_{\rm HL}$ are determined by the navigation policy designed in the high-level task planner, and $v_{{\rm apex},c}$, $\Delta y_{1,c}$ and $\Delta y_{2,c}$ are fixed from the previous step, we manipulate $v_{{\rm apex},n}$ to adjust the sagittal phase-space plan and subsequently the lateral phase-space plan through the updated walking step timing. To this end, we sample a set of equidistant values $v_{{\rm apex},n} \in [v_{{\rm apex, min}},v_{{\rm apex, max}}]$ and calculate a cost $\lambda$, which penalizes deviation of $\Delta y_{1,n}$ and $\Delta y_{2,n}$ from their respective desired values $\Delta y_{1,d} \in \mathcal{R}_{\Delta y_1}$ and $\Delta y_{2,d} \in \mathcal{R}_{\Delta y_2}$\footnote{ $\Delta y_{1,d}$ and $\Delta y_{2,d}$ are heuristically selected according to our bipedal robot's leg kinematics. Exact values of $\Delta y_{1,d}$ and $\Delta y_{2,d}$ are shown in Table.~\ref{tab:value_table} in Sec.\ref{subsec:results_nominal}}. After the sampling, we set $v_{{\rm apex},n}$ to the optimal next apex velocity $v_{{\rm apex,opt}}$ that results in the minimum cost. This procedure is presented in Algorithm~\ref{alg:searching}.

\begin{algorithm}
\SetAlgoLined
 \textbf{Input:} $d$, $v_{{\rm apex},c}$, $\Delta y_{1,c}$, $\Delta y_{2,c}$, and a velocity sampling increment $v_{\rm inc}$\;
 \textbf{Set:} $v_{{\rm apex},n} \leftarrow v_{{\rm apex, min}}$, cost $\lambda \leftarrow \infty$, $\Delta y_{1, d}$ and $\Delta y_{2, d}$, desired step duration $T_d$, desired step width $W_{d}$ and cost weights $c_1$, $c_2$, $c_3$ and $c_4$\;
 \While{$v_{{\rm apex},n} \leq v_{{\rm apex, max}}$}{
  $t_{\rm FHWS}$, $t_{\rm SHWS}$ $\leftarrow$ sagittal PSP with ($d$, $v_{{\rm apex},c}$, $v_{{\rm apex},n}$)\;
  $\Delta y_{1,n}$, $\Delta y_{2,n}$ $\leftarrow$ Newton-Raphson Search \cite{zhao2017robust}\;
   $\lambda_{\rm new} = c_{1}| \Delta y_{1, d}-\Delta y_{1, n}| + c_{2}|\Delta y_{2, d}-\Delta y_{2, n}|$
    $ + c_{3}| T_{d} -(t_{\rm FHWS} + t_{\rm SHWS})|$ 
    $ + c_{4}|W_{d} - |y_{{\rm foot},n} - y_{{\rm foot},c}||$\;
   \If{$\lambda_{\rm new} < \lambda$}{
   $ \lambda \leftarrow \lambda_{\rm new}$\;
   $v_{{\rm apex, opt}} \leftarrow v_{{\rm apex}, n}$
   }
  
   $v_{{\rm apex},n} \leftarrow v_{{\rm apex},n} + v_{\rm inc}$
 }
 
 \textbf{Output:} $v_{{\rm apex},n}=v_{{\rm apex,opt}}$
 \caption{Optimal Next Apex Velocity Design for Lateral Waypoint Tracking}
 \label{alg:searching}
\end{algorithm}

Algorithm~\ref{alg:searching} also includes two regularization costs on step duration $t_{\rm FHWS} + t_{\rm SHWS}$ and step width $|y_{{\rm foot},n} - y_{{\rm foot},c}|$, respectively, where $T_d$ and $W_{d}$ are empirically selected to guarantee hardware implementation feasibility of the Digit robot (e.g., constraints from robot leg dynamics and kinematics)\footnote{Long step duration $>0.7$ s and large step width $>0.55$ m are impractical for maintaining Digit's locomotion safety.}. Algorithm~\ref{alg:searching} is robust to different step lengths during straight walking, however waypoint tracking during a turning sequence is more complex. 
In extreme turning cases where Algorithm~\ref{alg:searching} fails to find an apex velocity that yields viable waypoint tracking in Proposition~\ref{thm:tracking_viability}, we will propose an \textit{online} replanning mechanism to adjust the waypoint (see Sec.~\ref{subsec:action_planner}). 

\section{Task Planning via Belief Abstraction}
\label{sec:task planner}
This section will expound the high-level task planning structure, consisting of global navigation and local action planners that employ LTL to achieve safe locomotion navigation in a partially observable environment with dynamic obstacles. Low-level locomotion dynamics constraints are encoded into LTL specifications to ensure that high-level actions can be successfully executed by the middle-level motion planner to maintain balancing safety. 

%
\begin{defn}[Navigation Safety]
Navigation safety is defined as safe maneuvering in partially observable environments with uneven terrain while avoiding collisions with static and dynamic obstacles.
\end{defn}

To achieve safe navigation, the task planner evaluates observed environmental events at each walking step and commands a safe action set to the middle-level motion planner as shown in Fig.~\ref{fig:frame} while guaranteeing goal positions to be visited \textit{in order} and \textit{infinitely often}. In particular, we study a pick-up and drop-off task while guaranteeing static and dynamic obstacle collision avoidance.

We design our task planner using formal synthesis methods to ensure locomotion actions guarantee navigation safety and liveness, specifically we use General Reactivity of Rank 1 (GR(1)), a fragment of LTL.
 GR(1) allows us to design temporal logic formulas ($\varphi$) with atomic propositions (AP ($\varphi$)) that can either be \textsf{True} ($\varphi \vee \neg\varphi$) or \textsf{False} ($\neg$\textsf{True}). With negation ($\neg$) and disjunction ($\vee$) one can also define the following operators: conjunction ($\wedge$), implication ($\Rightarrow$), and equivalence ($\Leftrightarrow$). Other temporal operators include ``next" ($\bigcirc$),
 ``eventually" ($\Diamond$), and ``always" ($\square$). Safety specifications capture how the system and environment may transition during one step of the synthesized controller's execution, while liveness specifications capture which transitions must happen \textit{infinitely often}. Further details of GR(1) can be found in \cite{gr1}. Our implementation uses the SLUGS reactive synthesis tool \cite{ehlers2016slugs} to design specifications with Atomic Propositions (APs), natural numbers, and infix notation, which are automatically converted to ones using only APs.

\begin{rem}
The discrete abstraction granularity required to plan walking actions for each keyframe is too fine to synthesize plans for large environment navigation. Therefore, we have split the task planner into two layers: A high-level navigation planner that plays a navigation and collision avoidance game against the environment on a global coarse discrete abstraction, and an action planner that plays a local game on a fine abstraction of the local environment (corresponding to one coarse cell). The action planner generates action sets at each keyframe to progress through the local environment and achieve the desired coarse-cell transition after multiple walking steps. 
\end{rem}

\subsection{Navigation Planner Design}
A top-down projection of the navigation environment is discretized into a coarse two-dimensional grid
as shown in Fig. \ref{fig:belief_results}. Each time the robot enters a new cell, the navigation planner evaluates the robot's discrete location ($l_{r,c} \in \mathcal{L}_{r,c}$) and heading ($h_{r,c} \in \mathcal{H}_{r,c}$) on the coarse grid, as well as the dynamic obstacle's location ($l_o \in \mathcal{L}_o$), and determines a desired navigation action ($n_a \in \mathcal{N}_a$). The planner can choose for the robot to stop, or to transition to any reachable safe adjacent cell. $\mathcal{L}_{r,c}$ and $\mathcal{L}_o$ denote sets of all coarse cells the robot and dynamic obstacle can occupy, while $\mathcal{H}_{r,c}$ represents the four cardinal directions in which the robot can travel on the coarse abstraction. The dynamic obstacle moves under the following assumptions: (a) it will not attempt to collide with the robot when the robot is standing still,
(b) its maximum speed only allows it to transition to an adjacent coarse cell during one turn of the navigation game, and
(c) it will eventually move out of the way to allow the robot to pass. Assumption (c) prevents a deadlock \cite{alonso2018reactive}. Static obstacle locations are encoded as safety specifications. Given these assumptions, the task planner in Section~\ref{subsec:task_planner_synth} will guarantee that the walking robot can achieve a specific navigation goal. 

\subsection{Action Planner Design}
\label{subsec:action_planner}
The local environment, i.e., one coarse cell, is further abstracted into a fine discretization.
At each walking step, the action planner evaluates the robot's state in the environment ($\boldsymbol{e}_{\rm HL}$)\footnote{We use the symbol $\boldsymbol{e}_{\rm HL}$ to represent the robot state, since this symbol represents the second player in the game, i.e., the environment player.} consisting of the discrete waypoint location ($l_{r,f} \in \mathcal{L}_{r,f}$) and heading ($h_{r,f} \in \mathcal{H}_{r,f}$) on the fine grid, as well as the robots current stance foot index ($i_{\rm st}$), and determines an appropriate action set ($\boldsymbol{a}_{\rm HL}$) defined in Def.~\ref{def:keyframe}.
The action planner generates a sequence of locomotion actions guaranteeing that the robot eventually transitions to the next desired coarse cell while ensuring all action sets are safe and achievable based on $\boldsymbol{e}_{\rm HL}$ and $\boldsymbol{a}_{\rm HL}$.
Note that, the fine abstraction also models the terrain height for each fine-level cell, allowing the action planner to choose the correct step height $\Delta z_{\rm foot}$ for each keyframe transition. 

During locomotion, the nominal robot state transitions are deterministically modeled within the action planner based on the current game state and system action, but the nominal transition is not guaranteed. To account for this, we model additional necessary non-deterministic transitions to handle the following cases:
\begin{itemize}

    \item Not all the robot states can be captured in the discrete abstraction, such as the robot CoM velocity, which, however, may still affect transitions, i.e. the robot's CoM deviates from the desired high-level waypoint.

   \item The robot may be perturbed externally while walking, altering the foot location at the next walking step.
\end{itemize}
We have encoded non-deterministic transitions, and associated transition flags ($t_{nd}$), to capture these cases into the action planner's environment assumptions. This flag variable $t_{nd}$  is encoded as a special automaton state that will be used to replan the foot location of the next walking step. An example of addressing a sagittal perturbation will be shown in Fig.~\ref{fig:sagittal_replan} (c).


An example of modeled non-deterministic transitions can be seen in Fig.~\ref{fig:non_determinstic}.  The CoM trajectory sometimes imperfectly tracks the waypoints due to accumulated differences in the continuous keyframe state represented by the same discrete state $\boldsymbol{e}_{\rm HL}$. The reduced-order motion planner identifies when the waypoint needs to be shifted from the lateral case and informs the action planner, which verifies the updated waypoint is allowed by the non-deterministic transition model and continues planning from the new waypoint.

\subsection{Encoding Low-level Dynamics Constraints into High-level Planner Specifications}
\label{subsec:low_in_high}
To ensure the action planner only commands safe and feasible actions, we must take into account the underlying \textit{Locomotion Safety}. This is achieved by capturing low-level constraints in the high-level planner specifications. Action planner state transition limitations based on straight walking step length constraints in Theorem~\ref{thm:straight}, and kinematic constraints from the bipedal robot leg, are directly encoded in the action planner specifications. 
\textit{Locomotion safety} is guaranteed when the combination of apex velocity, heading angle change, and foot placement meets Theorems~\ref{thm:straight}-\ref{thm:steering1}. These constraints are not able to be directly captured as the action planner does not reason about CoM velocity and the dynamic equations of motion can not be encoded in symbolic specifications. Instead, they are captured by generating a library of permissible turning sequences based on discrete robot states that are known to meet the above constraints (see Table.~\ref{tab:value_table}). For example, given $\omega = 3.15$ rad/s, $\Delta y_{2,c} = 0.14$ m (equals to $\Delta y_{2,d}$ in Algorithm~\ref{alg:searching}), and an allowable $v_{\rm apex}$ range $[0.2, 0.7]$ m/s, Theorem~\ref{thm:steering1} results in $\Delta\theta \leq 24.40^\circ$. Any turning angle larger than this value will result in a high-level action that is not executable by the middle-level motion planner. Thus we choose $\Delta\theta = 22.5 ^\circ$ such that we can complete a $90^\circ$ turn in $4$ consecutive walking steps. 
A safe turning sequence can be seen in Fig.~\ref{fig:non_determinstic}. 

To ensure that collision avoidance in the abstract game translates to collision-free locomotion in the continuous domain, we guarantee the location $l_{r,f}$ stays far enough away from any obstacles.
Algorithm.~\ref{alg:searching} ensures that 
the distance between $l_{r,f}$ and the robot's desired foot placement does not exceed $b_{\rm safety}$ as detailed in Sec.~\ref{sec:Tracking}. 
The action planner guarantees $l_{r,f}$ is never in a cell that is less than a distance $b_{\rm safety}$ away from the neighboring coarse cell that may contain static or dynamic obstacles via safety specifications. The planner guarantees this distance even after non-deterministic sagittal and lateral transitions, ensuring collision avoidance.

\begin{rem}
The navigation planner cell can not be an arbitrary size because, in a locomotion setting, the underlying dynamics of the bipedal system and multiple walking steps need to be considered to ensure the safe and correct transition between adjacent coarse cells. 
\end{rem}

\begin{figure}[t]
\centerline{\includegraphics[width=.45\textwidth]{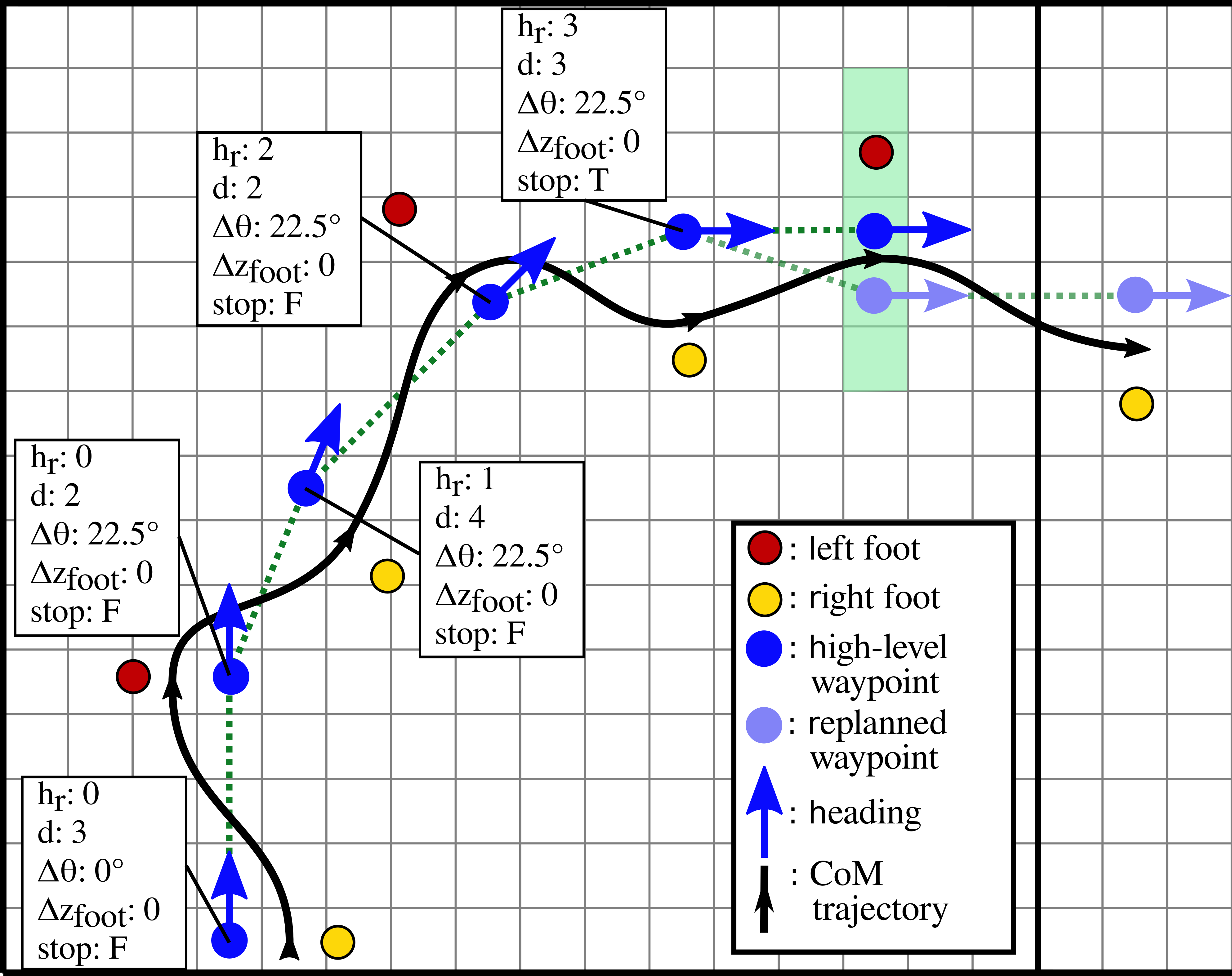}}
\caption{Illustration of fine-level steering walking within one coarse cell.
Discrete actions are planned at each keyframe allowing the robot to traverse the fine grid toward the next coarse cell. The waypoint transitions non-deterministically following the turn. A set of locomotion keyframe decisions are also annotated.}
\label{fig:non_determinstic}
\vspace{-0.15in}
\end{figure}

\subsection{Task Planner Synthesis}
\label{subsec:task_planner_synth}

A navigation game structure is proposed by including robot actions in the tuple $\mathcal{G} := (\mathcal{S}, s^{\rm init}, \mathcal{T}_{N})$ with
\begin{itemize}
    \item $\mathcal{S} = \mathcal{L}_{r,c} \times \mathcal{L}_o \times \mathcal{H}_{r,c} \times \mathcal{N}_a$ is the augmented state;
    \item $s^{\rm init} = (l_{r,c}^{\rm init},l_o^{\rm init},h_{r,c}^{\rm init},n_a^{\rm init}) $ is the initial state;
    \item $\mathcal{T}_{N} \subseteq \mathcal{S} \times \mathcal{S}$ is a transition relation describing the possible moves of the robot and the obstacle.
\end{itemize}
To synthesize the transition system $\mathcal{T}_{N}$, we define the rules for the possible successor state locations which will be further expressed in the form of LTL specifications $\psi$. The successor location of the robot is based on its current state and action $succ_r(l_{r,c},h_{r,c},n_a) = \{l'_{r,c} \in \mathcal{L}_{r,c} | \exists l'_o, h'_{r,c} ((l_{r,c},l_o,h_{r,c},n_a),(l'_{r,c},l'_o,h'_{r,c},n_a'))\in \mathcal{T}_{N}\}$. We define the set of possible successor robot actions at the next step as $succ_{n_a}(n_a,l_{r,c},l_{r,c}',l_o,l'_o,h_{r,c},h_{r,c}') = \{ n_a' \in \mathcal{N}_a  | ((l_{r,c},l_o,h_{r,c},n_a) , (l'_{r,c},l'_o,h'_{r,c},n_a')) \in \mathcal{T}_{N} \}$. We define the set of successor locations of the obstacle. $succ_o(l_{r,c},l_o,n_a) = \{l'_o \in \mathcal{L}_o | \exists l'_{r,c}, h'_{r,c}. ((l_{r,c},l_o,h_{r,c},n_a) , (l'_{r,c},l'_o,h'_{r,c},n_a')) \in \mathcal{T}_{N} \}$.
Later we will use a belief abstraction inspired by \cite{bharadwaj2018synthesis} to solve our synthesis in a partially observable environment.

The task planner models the robot and environment interplay as a two-player game. The robot action is Player 1 while the possible adversarial obstacle is Player 2. The synthesized strategy guarantees that the robot will always win the game by solving the following reactive problem.

\noindent\textbf{Reactive synthesis problem:}
Given a transition system $\mathcal{T}_{N}$ and linear temporal logic specifications $\psi$, synthesize a winning strategy for the robot such that only correct decisions are generated in the sense that the executions satisfy $\psi$.

The action planner is synthesized using the same game structure as the navigation planner, with possible states and actions corresponding to Section~\ref{subsec:action_planner}. Non-deterministic robot location transitions are captured in the robot successor function $succ_{r,f}(l_{r,f},h_{r,f},\boldsymbol{a}_{\rm HL}) = \{l'_{r,f} \in \mathcal{L}_{r,f}, h'_{r,f} \in \mathcal{H}_{r,f} | ((l_{r,f},h_{r,f},\boldsymbol{a}_{\rm HL}),(l'_{r,f},h'_{r,f},\boldsymbol{a}_{\rm HL}'))\in \mathcal{T}_{A}\}$, where $\mathcal{T}_{A}$ is the transition relation in the action planner. Compared to the transition relation $\mathcal{T}_{N}$, $\mathcal{T}_{A}$ does not have the obstacle location $l_o$ but includes locomotion actions $\boldsymbol{a}_{\rm HL}$.  Given the current robot state and action, $succ_{r,f}$ provides a set of possible locations at the next turn in the game. Obstacle avoidance is taken care of in the navigation game the obstacle location $\mathcal{L}_o$ and successor function $succ_o$ are not needed for action planner synthesis. Since reactive synthesis is used for both navigation and action planners, and the action planner guarantees the robot transition in the navigation game, the correctness of this hierarchical task planner is guaranteed.

\subsection{Belief Space Planning in A Partially Observable Environment}
The navigation planner above synthesizes a safe game strategy that is always winning but only in a fully observable environment. 
We relax this assumption by assigning the robot a visible range only within which the robot can accurately identify a dynamic obstacle's location. To reason about where an out-of-sight obstacle could be, we devise an abstract belief set construction method based on the work in \cite{bharadwaj2018synthesis}. Using this belief abstraction, we explicitly track the possible discrete locations of a dynamic obstacle, rather than assuming it could be in any non-visible cell. The abstraction is designed by partitioning regions of the environment into sets of discrete belief regions ($R_b$) and constructing a powerset of these regions ($\mathcal{P}(R_b)$). We choose smaller partitions around static obstacles that may block the robot's view as this allows the planner to guarantee collision-free navigation for a longer horizon like the scenario depicted in Fig. \ref{fig:belief_trans}.
We index each set in $\mathcal{P}(R_b)$ to represent a belief state $b_o \in \mathcal{B}_o$ that captures non-visible regions potentially with a dynamic obstacle.

The fully observable navigation game structure is modified to generate a partially observable belief-based navigation game with an updated state $\mathcal{S}_{{\rm belief}}$ and transition system $\mathcal{T}_{\rm belief}$ In addition to the obstacle location $l_o \in \mathcal{L}_o$, $\mathcal{S}_{{\rm belief}}$ captures the robot's belief of the obstacle $b_o \in \mathcal{B}_o$. A visibility function $vis : \mathcal{S}_{{\rm belief}} \rightarrow \mathbb{B} $ is added such that it maps the state ($l_{r,c}, l_o$) to the Boolean $\mathbb{b}$ as \textsf{True} if and only if $l_o$ is a location in the visible range of $l_{r,c}$. We do not need to modify $succ_{n_a}$ since the dynamic obstacle only affects the possible one-step robot action if it is in the visible range. $succ_r$ also remains the same as the relationship between the robot's actions and its state is not changed by the belief. The set of possible successor beliefs of the obstacle location, $b'_o$, is defined as $succ_{b_o} =\{b'_o \in \mathcal{B}_o | ((l_{r,c},b_o,l_o) , (l'_{r,c},b'_o,l_o)) \in \mathcal{T}_{\rm belief} \}$ where $b'_o$ indexes $\emptyset$ 
when $vis(l_{r,c},l'_o) = \textsf{True}$ and $b'_o$ indexes a nonempty set in $\mathcal{P}(R_b)$ when $vis(l_{r,c},l'_o) = \textsf{False}$.

\begin{figure}[t]
\centering
\begin{subfigure}{.68\columnwidth}
  \includegraphics[width=\linewidth]{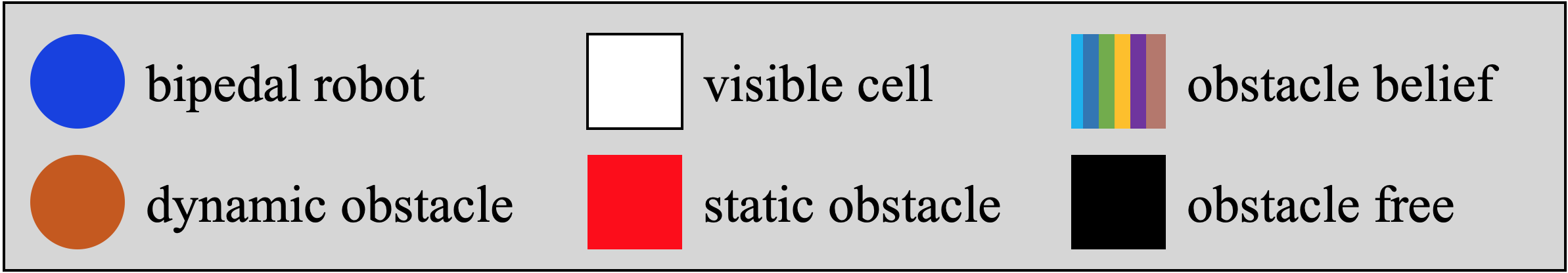}
  \label{fig:Belief_trans1}
\end{subfigure}%

\begin{subfigure}{.5\columnwidth}
  \centering
  \includegraphics[width=.95\linewidth]{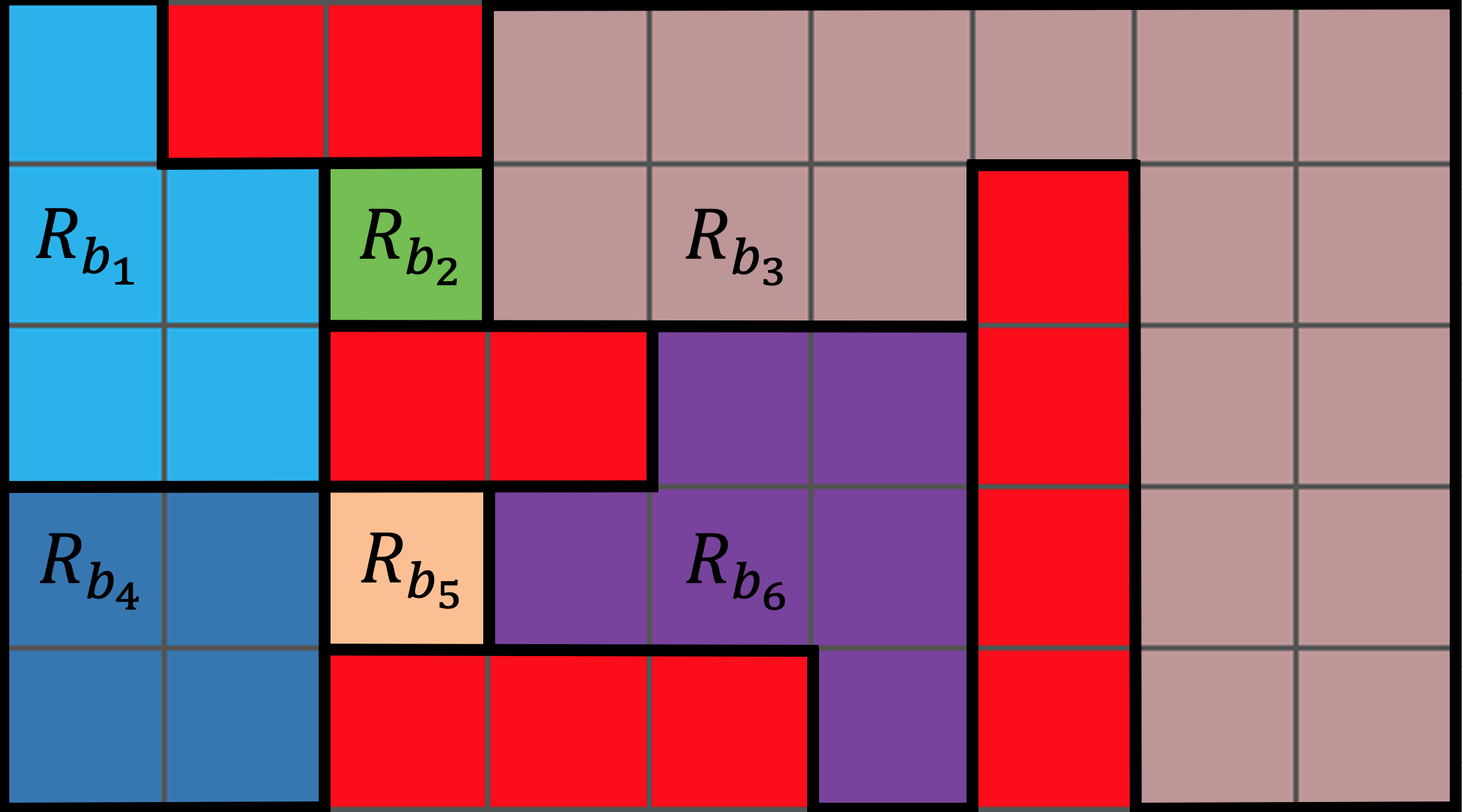}
  \caption{Environment divided into belief regions}
  \label{fig:Belief_trans1}
\end{subfigure}%
\begin{subfigure}{.5\columnwidth}
  \centering
  \includegraphics[width=.95\linewidth]{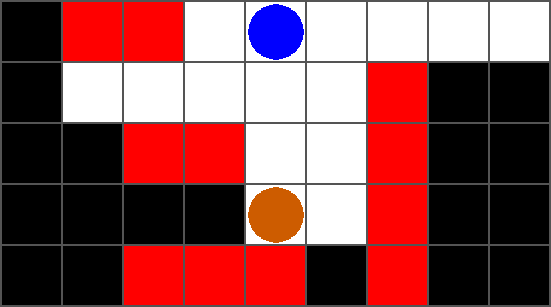}
  \caption{Obstacle before leaving visible range}
  \label{fig:Belief_trans2}
\end{subfigure}
\par\medskip
\begin{subfigure}{.5\columnwidth}
  \centering
  \includegraphics[width=.95\linewidth]{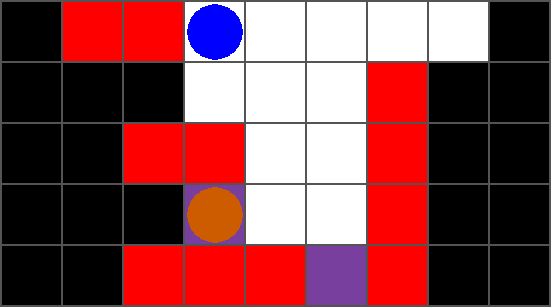}
  \caption{Obstacle not visible to robot}
  \label{fig:Belief_trans3}
\end{subfigure}%
\begin{subfigure}{.5\columnwidth}
  \centering
  \includegraphics[width=.95\linewidth]{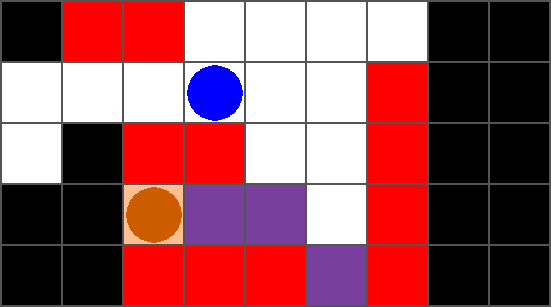}
  \caption{Obstacle not visible to robot}
  \label{fig:Belief_trans4}
\end{subfigure}
\par\medskip
\begin{subfigure}{.5\columnwidth}
  \centering
  \includegraphics[width=.95\linewidth]{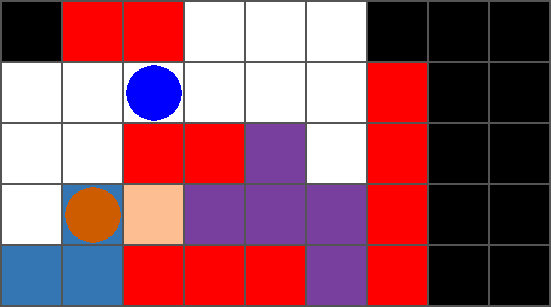}
  \caption{Obstacle not visible to robot}
  \label{fig:Belief_trans5}
\end{subfigure}%
\begin{subfigure}{.5\columnwidth}
  \centering
  \includegraphics[width=.95\linewidth]{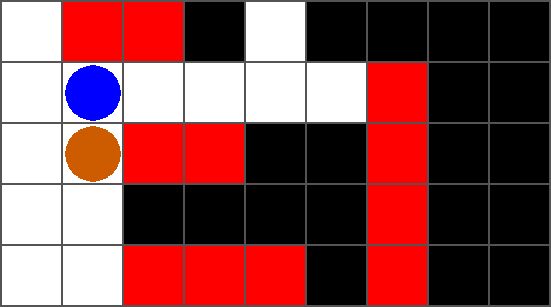}
  \caption{Obstacle reappears in visible range}
  \label{fig:Belief_trans6}
\end{subfigure}
\caption{Simulation showing how the navigation planner's belief evolves when the dynamic obstacle leaves the visible range for several turns. 6 colored belief regions are shown, as well as the robot (blue circle), the dynamic obstacle (orange circle), and the static obstacles (red cells). Black cells represent non-visible cells believed to be obstacle-free while white cells are visible. The planner believes the obstacle could be in any colored cell depicted, and can therefore reason where the obstacle could and could not reappear, allowing the planner to determine which navigation actions are safe.
}
\label{fig:belief_trans}
\vspace{-0.2in}
\end{figure}

Four classes of belief transitions, shown in Fig. \ref{fig:belief_trans}, are defined for accurate and meaningful belief tracking:
\begin{itemize}
   \item Visible to visible: as in the fully observable case, the obstacle may transition to any adjacent visible cell.
   \item Visible to belief: the belief state represents the set of regions containing non-visible cells adjacent to the obstacle's previous visible location.
   \item Belief to belief: the obstacle could be in any non-visible (or newly visible\footnote{Due to the turn-based nature of the game, an obstacle may be in the robot's visible range after the robot makes a move, but the obstacle may move from this newly visible cell before the robot reevaluates its new visible range.
   }) cell represented by the current belief state, the next belief state represents the current belief plus the belief regions the dynamic obstacle could have entered given its limited motion capability. 
   \item Belief to visible: similar to the previous case, the current obstacle may be in any non-visible or newly visible cell represented by the planner's belief, and may move to any adjacent cell, which defines the visible cells it could appear in at the next time step.

\end{itemize}

This method of belief tracking guarantees that all real transitions the obstacle can make during its turn are captured in the planner's belief. When the obstacle enters cells in a new belief region, the planner believes it could be anywhere in that region, therefore the belief is an over-approximation of possible obstacle locations. We guarantee that the obstacle is within the regions captured by the belief state, therefore we can guarantee that the obstacle can only appear in
a visible cell when there is a modeled transition from the current belief state to that cell.

Since both the action planner and the allowable navigation actions remain the same for the partially observable game, 
the game captures the same safety guarantees,
but allow for a larger set of navigation options than would be possible without tracking the belief of the dynamic obstacle's location. 




\subsection{Belief Tracking of Multiple Obstacles}
\label{subsec:belief_multi}
Our task planner is extensible to environments with multiple dynamic obstacles. 
It is possible to directly add any number of additional obstacles and their associated beliefs to the navigation planning game, however, the synthesis has polynomial time complexity. To improve computational tractability, we merge all non-visible obstacles' believed states into one combined belief region. Reasoning about a combined belief region still allows the planner to guarantee collision-free navigation without the complexity of tracking each obstacle individually.


To model a combined belief state we separate the obstacles' state from its belief. Each obstacle's state is either a visible cell on the grid or an index representing the obstacle is not visible ($l_{o,i,c}\in{L}_{o,i,c} | {L}_{o,i,c} = \mathcal{L}_{o}+\mathcal{I}_{nv}$). The joint belief state consists of the powerset of belief regions, including the empty set when all obstacles are visible. ($b_{oj}\in\mathcal{B} | \mathcal{B} = \mathcal{P}(R_b)$).


We generate a new multi-obstacle game structure $\mathcal{G}_{\rm combined-belief} := (\mathcal{S}_{\rm belief}, s_{\rm belief}^{\rm init}, \mathcal{T}_{\rm belief}, vis)$ with
\begin{itemize}
    \item $\mathcal{S}_{\rm belief} = \mathcal{L}_{r,c} \times \mathcal{L}_{o,i,c} \times \mathcal{B}_o \times \mathcal{H}_{r,c} \times \mathcal{N}_a$;
    \item $s_{\rm belief}^{\rm init} = (l_{r,c}^{\rm init},l_{o,i,c}^{\rm init},\{b_o^{\rm init}\},h_{r,c}^{\rm init},n_a^{\rm init})$ is the initial location of the obstacle known \textit{a priori};
    \item $\mathcal{T}_{\rm belief} \subseteq  \mathcal{S}_{\rm belief} \times \mathcal{S}_{\rm belief}$ are possible transitions where $((l_{r,c},l_{o,i,c},b_o,h_{r,c},n_a),(l'_{r,c},l'_{o,i,c},b'_o,h'_{r,c},n'_a)) \in \mathcal{T}_{\rm belief}$;
    \item $vis : \mathcal{S}_{{\rm belief}} \rightarrow \mathbb{B} $ is a visibility function that maps the state ($l_{r,c}, l_{o,i}$) to the boolean $\mathbb{b}$ as \textsf{True} iff $l_{o,i}$ is a real location in the visible range of $l_{r,c}$.
\end{itemize}
This game requires new specifications that govern $succ_{o,i}(l_{r,c},l_{o,i,c},b)$ and $succ_{b}(l_{r,c},l_{o,i,c},b)$, the allowable successor obstacle state and joint belief state, all other successor functions remain the same.
Even though the belief can represent multiple obstacles, the possible belief-to-belief transitions are the same as when the belief state represents a single obstacle. 
The key specifications to be changed are those governing $succ_{b_o}$ when an obstacle enters or exits the visible range. These changes can be made in the specifications defining the successor belief state $succ_{b_o}$. 


\section{Safe Recoverability and Replanning}
\label{sec:recoverbility}
\begin{figure}[t]
\centerline{\includegraphics[width=.49\textwidth]{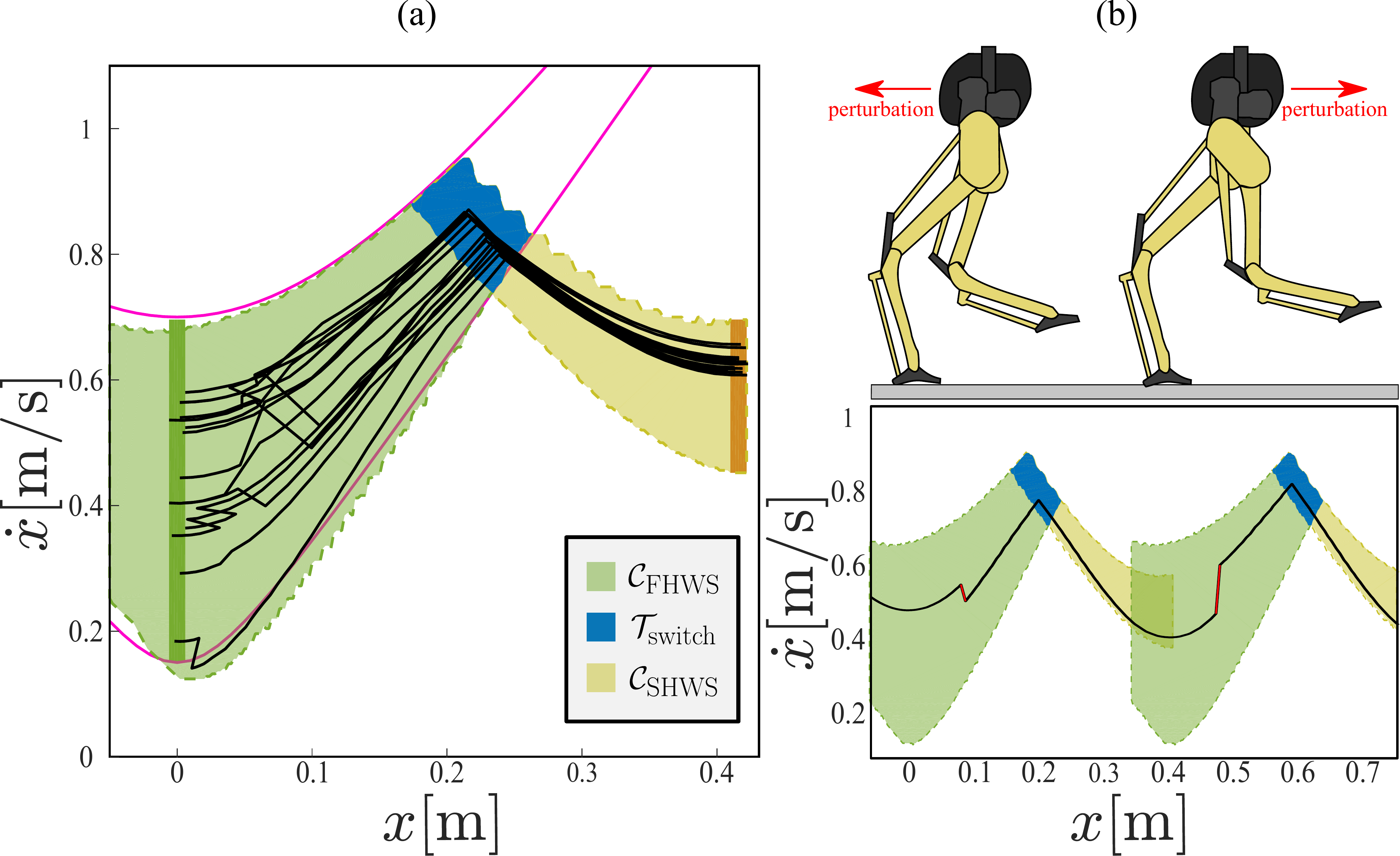}}
\caption{Results of OWS robust PSP. (a) shows 15 random keyframe transitions with bounded disturbances $\boldsymbol{\tilde \Xi}_{\rm execution}$, where $\mathcal{T}_{\rm OWS} = (0.416 \textnormal{ m}, [0.45, 0.7] \textnormal{ m/s})$. (b) Composition of controllable regions of OWS.  Here, we demonstrate that the synthesized controller is able to handle the perturbed CoM trajectory, shown as a black solid line, inside the superimposed controllable regions and successfully complete multiple steps when controllable regions are composed as proposed in Corollary~\ref{col:composition}.}
\label{fig:perturbation_rocs}
\end{figure}

The proposed sequential composition of controllable regions and reachability analysis in Sec.~\ref{subsubsec:capturebasin} allows our middle-level motion planner to be robust against perturbations exerted on the CoM in the sagittal space. Given a keyframe transition for OWS, the synthesized controller is able to guarantee that the CoM state reaches the targeted state within OWS, thus successfully completing an OWS safely. In Fig.~\ref{fig:perturbation_rocs}(b) we show the composition of controllable regions for multiple walking steps and demonstrate that the CoM trajectory is recoverable when employing the synthesized controller. Table.~\ref{tab:rocs_table} shows the success rate for randomly generated keyframe transitions, where the step length is $d_1 = 0.312$ m, $d_2 = 0.416$ m and $d_3 = 0.52$ m. The data is generated using ROCS \cite{li2018rocs} with $1000$ runs for each desired keyframe transition, a randomly selected $\boldsymbol{\xi}_c \in \boldsymbol{\Xi}_c$ and 
the applied disturbance bound at execution $\boldsymbol{\tilde \Xi}_{\rm execution}$\footnote{During online execution, the applied disturbance $\boldsymbol{\tilde \xi}_{\rm execution}$ is an instantaneous jump in the system states and is significantly larger than the considered disturbance during synthesis $\boldsymbol{\tilde \Xi}_{\rm synthesis}$ of the controllable regions.} is uniformly distributed within [$-2, 2$] m for CoM position and [$-5, 5$] m/s for CoM velocity.
The controllable regions are synthesized with state space granularity of $(0.002$ m, $0.004$ m/s), a control input $\omega \in [2.8,3.5]$ rad/s with a granularity of $0.02$ rad/s, 
and the added noise bound at synthesis $\boldsymbol{\tilde \Xi}_{\rm synthesis}$ is uniformly distributed within [$-0.01, 0.01$] m for CoM position and [$-0.02, 0.02$] m/s for CoM velocity. In Fig.~\ref{fig:perturbation_rocs}(a), we show $15$ successful random keyframe transitions where $v_{{\rm apex},n}=[0.45, 0.7]$ m/s and $d = 0.415$ m. The offline synthesis of the controller and controllable regions of a single transition as described in Table.~\ref{tab:rocs_table} takes on average 3.6 seconds.

 Large perturbations can push the system state outside of the controllable regions and the synthesized controller cannot recover to $\mathcal{T}_{\rm switch}$. To safely recover from such large perturbations, we employ a variant of the capture point formulation \cite{zhao2017robust,englsberger2011bipedal} to redesign the next foot position $x_{{\rm foot},n}$ while maintaining the desired $v_{{\rm apex},n}$ via the following formula:
\begin{equation}
    x_{{\rm foot},n} = x_{\rm switch} + \frac{1}{\omega}(\dot{x}_{\rm switch,dist}^{2} + v_{{\rm apex},n}^{2})^{1/2}
\end{equation}
where $x_{\rm switch}$ is determined analytically based on the nominal transition, and $\dot{x}_{\rm switch,dist}$ is the post-disturbance sagittal CoM velocity at switch instant and computed through a position guard $x = x_{\rm switch}$ shown as the vertical dashed line in Fig.\ref{fig:sagittal_replan} (a). The nominal foot position is determined by the high-level waypoint. In case that the new foot location lands in a different fine cell, the \textit{online} integration mechanism between the high level and middle level will update the action planner for a new waypoint location as shown in Fig.~\ref{fig:sagittal_replan}(b) and (c).
The action planner reacts to the perturbation by replanning $d$ and $\Delta\theta$ in $\boldsymbol{a}_{\rm HL}$, which further induces a waypoint change at the next walking step. In particular, the non-deterministic transition flag $t_{nd} = \{ \textsf{nominal}, \textsf{forward}, \textsf{backward}\}$ indicates the perturbation direction. The automata shown in Fig. \ref{fig:sagittal_replan} (c) is a fragment from the larger action planner consisting of 21447 nodes. The navigation planner automaton has 20545 nodes. Online resynthesis of these planner automata is computationally intractable, and thus we incorporate the non-deterministic transition flag $t_{nd}$ into the automaton offline synthesis and employ them online for action replanning.

\begin{rem}
It is common sense that any robotic system can not handle arbitrary large perturbations (due to limited actuation, control capability, kinematics limits, etc). Here, we merely demonstrate that our replanning strategy (i.e., the nondeterministic transitions in the action planner in Sec.\ref{subsec:action_planner}) has the potential to handle certain large perturbations such that the CoM state is pushed outside the controllable region. Certainly, there is no formal guarantee of recovery from extremely large perturbations. Note that, the recovery is formally guaranteed when the perturbed CoM is within the controllable region.
\end{rem}

\begin{table}[t]
    \centering
    \caption{Success rate of perturbed OWS transitions}
\begin{tabular}{ |c|c|c|c|  }
 \hline
 \multirow{2}{*}{$v_{{\rm apex},n}$ margin}&  \multicolumn{3}{c|}{success rate}\\
\cline{2-4} & $d_1$ & $d_2$ & $d_3$ \\
\hline
$[0.2, 0.45]$ m/s & $90.2\%$ & $91.6\%$ & $92.5 \%$ \\
$[0.45, 0.7]$ m/s & $91.8 \%$ & $ 92.2 \%$ & $93.6 \%$\\
\hline
\end{tabular}
\label{tab:rocs_table}
\end{table}
\vspace{-0.15in}

\begin{figure}[ht!]
\begin{subfigure}{1\columnwidth}
\centerline{\includegraphics[width=.98\textwidth]{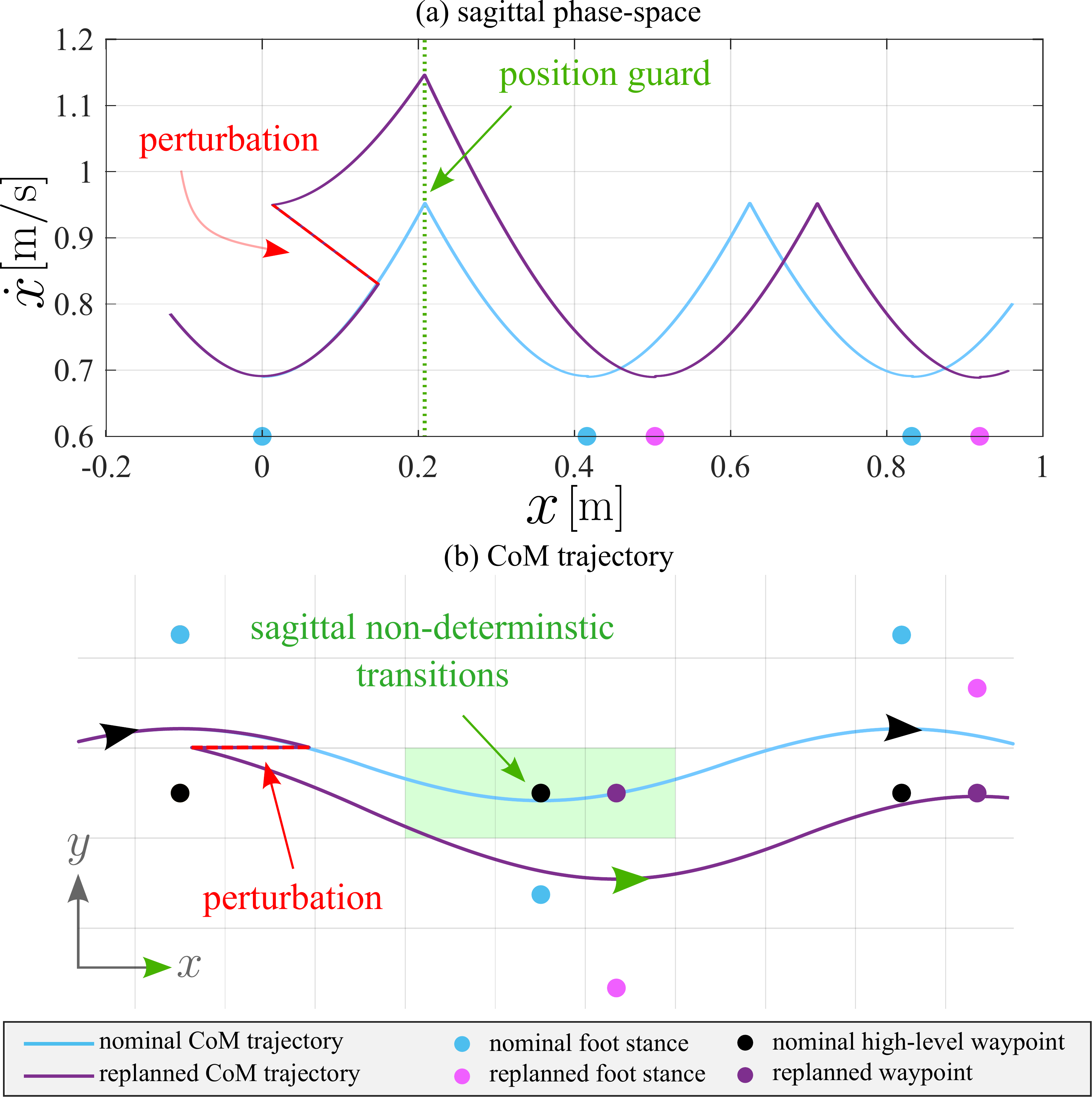}}
\caption*{(c) Navigation Automaton Fragment with Replanning}
\end{subfigure}%
\par\smallskip
\centering
\begin{subfigure}{1\columnwidth}
  \includegraphics[width=\linewidth]{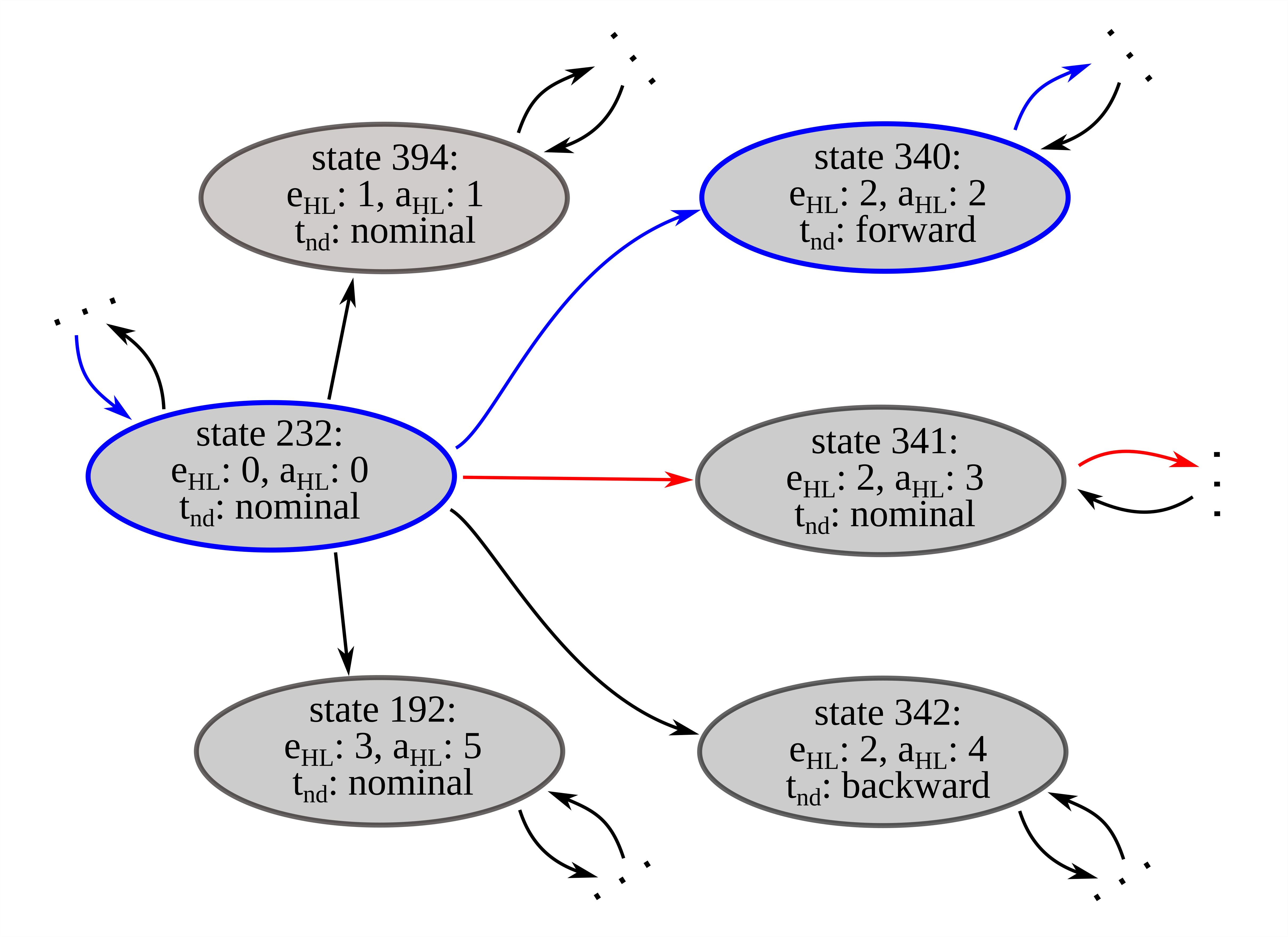}
  \label{fig:automaton}
\end{subfigure}%
\vspace{-0.15in}
\caption{Safe recovery from a large perturbation. (a) shows the sagittal phase-space plan, where a position guard is used to determine a safe replanned foot location to recover from the perturbation. (b) shows the CoM trajectory in Cartesian space and the \textit{online} integration of the high-level action planner and the middle-level PSP for a waypoint modification.
(c) shows a fragment of the synthesized action planner automaton capturing modeled non-deterministic transitions (with the associated flag $t_{nd}$). For each next state of the environment ($e_{\rm HL}$), there is a set of game states corresponding to all possible $t_{nd}$. Blue transitions capture the replanned execution when the robot CoM is perturbed forward while red transitions depict a nominal execution without any perturbation. Numerical values for $e_{\rm HL}$ and $a_{\rm HL}$ index distinct environment state and robot action sets in the algorithm implementation.
}
\label{fig:sagittal_replan}
\vspace{-0.05in}
\end{figure}

\vspace{0.1in}
\section{Low-level Control Implementation}
\label{sec:hardware}
In this section, we design a low-level full-body-dynamics-based controller to track the motion plan from the TAMP on the Digit bipedal robot \cite{agility}. The low-level hardware controller incorporates the angular-momentum-based foot placement controller\footnote{This controller modifies the foot placement from the phase-space plan to improve velocity tracking performance and achieve safer maneuvering.}~\cite{Gong2022AngularMomentum} and the passivity-based controller~\cite{sadeghian2017passivity} with modifications to handle non-periodic motion plans.

\subsection{Non-periodic Angular-Momentum-based Foot Placement Controller}

For low-level online execution, we build a foot placement controller based on the angular-momentum linear inverted pendulum (ALIP)~\cite{Gong2022AngularMomentum}, with a few critical modifications to track non-periodic phase-space plans. The motion plans between the ALIP controller \cite{Gong2022AngularMomentum} and our PSP differ in three folds: (i) state discretization; (ii) step duration; and (iii) the coordinate reference frame. Therefore, we adapt our phase-space plan and modify the ALIP controller to bridge these gaps.

First, the ALIP controller discretizes the ROM motion plan at $\boldsymbol{\xi}_{\rm switch}$, while our PSP uses $\boldsymbol{\xi}_{\rm apex}$ as keyframes. The phase-space plan between two consecutive keyframes is further transformed into an equivalent switching state, so it can be used by the ALIP controller. 
In the lateral plane, the ALIP controller takes a desired lateral velocity based on a periodic gait with fixed desired step width, whereas our phase-space plan has varying step widths and lateral velocities. We extend the ALIP controller to take in non-periodic lateral target velocities from the phase-space plan.

Another assumption in~\cite{Gong2022AngularMomentum} is that each step has a constant duration $T_{\rm ALIP}$ between $\boldsymbol{\xi}_{{\rm switch},c}$ and $\boldsymbol{\xi}_{{\rm switch},n}$. The steps in PSP, however, are non-periodic, and the step time $T_{\rm PSP}$ is the duration between $\boldsymbol{\xi}_{{\rm apex},c}$ and $\boldsymbol{\xi}_{{\rm apex},n}$. To command non-periodic phase-space plans to the ALIP controller, we relax the ALIP controller to take a varying step time $T_{{\rm ALIP},c} = t_{{\rm SHWS},p} + t_{{\rm FHWS},c}$ as seen in Fig.~\ref{fig:PSP2FP}. Note that two consecutive steps may have the same step time, i.e., $T_{{\rm ALIP},c} = T_{{\rm ALIP},n}$ 
can be true even in non-periodic walking.

In Sec.~\ref{subsubsec:safetyprop}, we propose a set of ROM-based safety theorems to generate safe turning behaviors (see Fig.~\ref{fig:steering_safety}). In these turning cases, however, the torso's heading direction is changing constantly and cannot be used as the reference frame. We instead adopt the PSP waypoint's heading direction as the reference frame and align the stance toe with that reference frame. As such, the target velocity in the next step needs a proper transformation to the reference frame of the current walking step. For example, the next-step trajectory from PSP is originally represented in that step's reference frame shown in red in Fig.~\ref{fig:PSP2FP}. To command safe turning, PSP hyperparameters (e.g., CoM velocities and foot placements) need to be transformed to the current reference frame shown in black in Fig.~\ref{fig:PSP2FP}. Fig.~\ref{fig:PSP2FP} shows the sagittal CoM phase-space trajectory for three consecutive keyframes, where $\dot{x}_{{\rm switch},n}$ is transformed from the next step's reference frame to the current one, as shown in red and black, respectively.

Although our hardware experiments are limited on even terrain walking (i.e., $\boldsymbol{p}_{\rm com,z}$ is fixed), our framework has demonstrated the non-periodic locomotion capability when tracking high-level waypoints and can be naturally extended to rough terrain maneuvering.


\begin{figure}[t]
\centerline{\includegraphics[width=.48\textwidth]{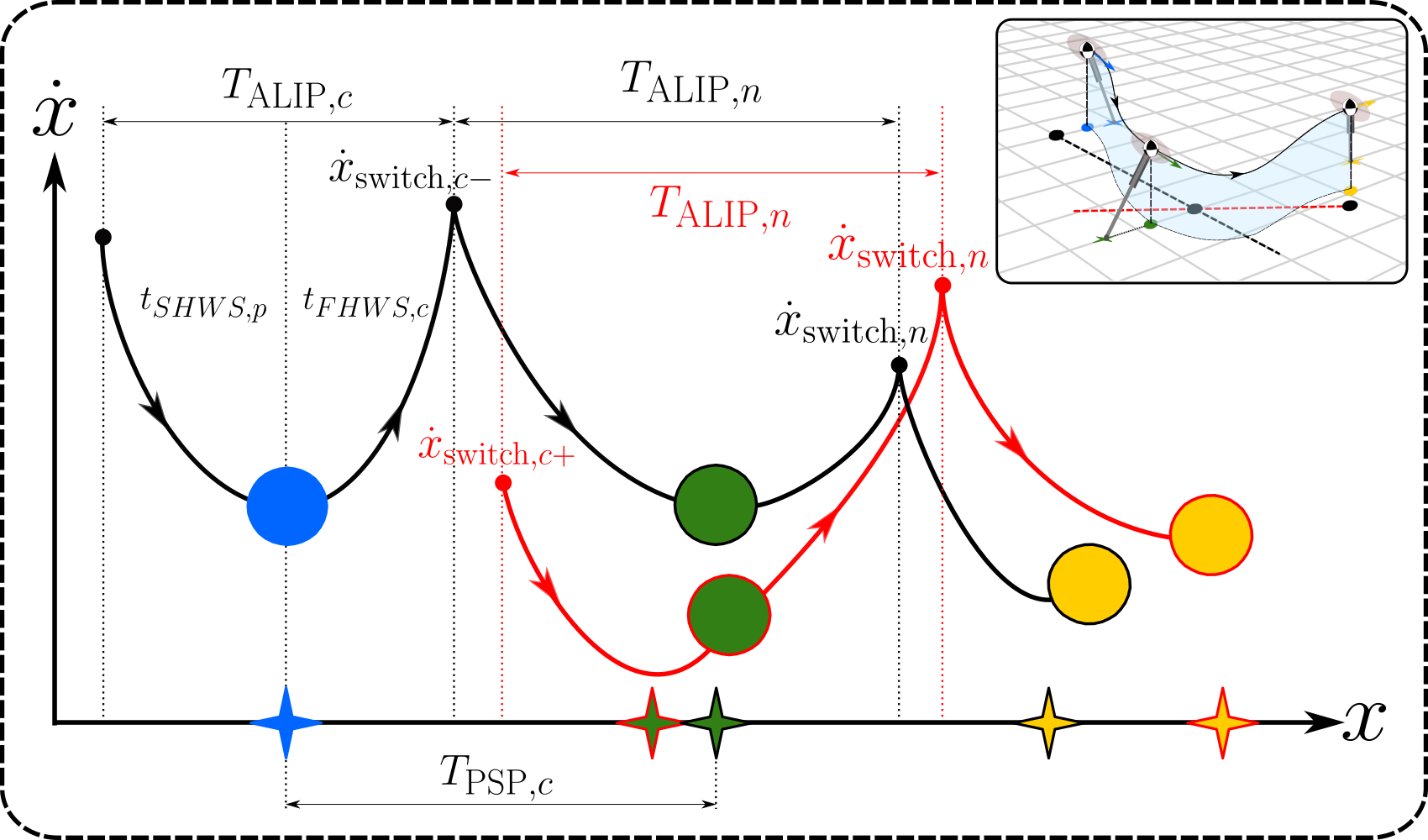}}
\caption{Transforming PSP to ALIP controller trajectory for the turning scenario. During the current walking step, the desired switching velocity needs to be transformed from the PSP representation \textcolor{red}{$\dot{x}_{{\rm switch},n}$} (in the next-step coordinate frame) to the current coordinate frame $\dot{x}_{{\rm switch},n}$. After contact at $\dot{x}_{{\rm switch},c}$, the coordinate reference is rotated according to the turn, and both the current switching velocity \textcolor{red}{$\dot{x}_{{\rm switch},c+}$} and \textcolor{red}{$\dot{x}_{{\rm switch},n}$} are represented in the next frame (red). In PSP, the coordinate transformation is executed at the apex state as explained in Sec.~\ref{subsubsec:safetyprop} and not switching state.}
\label{fig:PSP2FP}
\vspace{-0.05in}
\end{figure}

\subsection{Passivity-based Controller}
\label{sec:passivity}


At the low level, we implement a passivity-based controller to execute the phase-space plan from the middle-level motion planner. First, we design a full-body reference trajectory for control. We have, from PSP, the foot placement location and the CoM position trajectory that comprise the hyperparameters of the reference trajectory for each walking step. A smooth curve constructed with a half-period cosine function connects the measured state at the beginning of a step and the desired state at the end. The curve is defined in the task space, and the state incorporates the relative transformation between the CoM and two feet. A set of geometry-based inverse kinematics functions construct the full-body reference trajectory online.

We adopt the passivity-based controller \cite{sadeghian2017passivity} to achieve an accurate tracking performance at the joint level. The passivity-based controller preserves the natural dynamics, which is more appealing compared to the input-output linearization technique \cite{IOL_grizzle} that cancels these dynamics. 
Our controller takes in the target acceleration $\boldsymbol{\ddot{q}}^{{\rm target}}(t) = \boldsymbol{\ddot{q}}^{d}(t) + \boldsymbol{k}_p \boldsymbol{q}^e(t) + \boldsymbol{k}_d \boldsymbol{\dot{q}}^e(t)$, where the $\boldsymbol{\ddot{q}}^{d}(t)$ is the desired joint acceleration from the full-body reference trajectory, $\boldsymbol{q}^e(t)$ and $\boldsymbol{\dot{q}}^e(t)$ are the joint-level error for position and velocity, and $\boldsymbol{k}_p$ and $\boldsymbol{k}_d$ are the joint-level PD gains. For Digit, the dimension of $\boldsymbol{q}$ is 28, including the 6-DoFs world-to-torso floating joint, two 4-DoFs arms, and two 7-DoFs legs. The actuator torque is calculated based on the full-order dynamics of the Digit robot.
The passivity-based controller achieves asymptotical tracking performance, i.e., the joint position error $\boldsymbol{q}^e(t) = \boldsymbol{q}^d(t) - \boldsymbol{q}(t)$ converges to zero asymptotically.

\begin{figure*}[t]
\centerline{\includegraphics[width=1\textwidth]{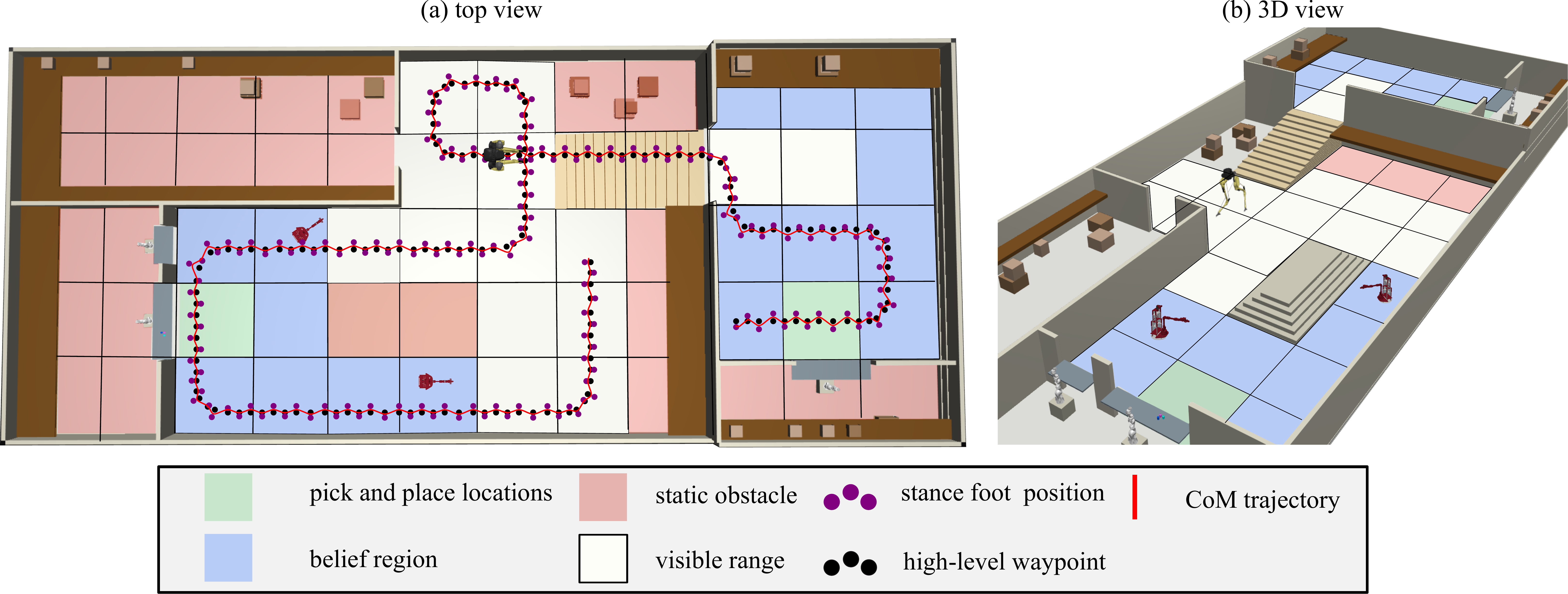}}
\caption{3D motion plans visualized using the Drake toolbox \cite{drake} on the Cassie robot navigating in the partially observable environment while avoiding collisions with two mobile robots that are treated as dynamic obstacles. Cassie's task is to move between designated initial and goal locations for package delivery. Trajectories of Cassie CoM, foot placements as well as environment coarse-level cell abstraction are shown in subfigure (a). Subfigure (b) shows a 3D view of the tested environment.}
\label{fig:psp}
\vspace{-0.1in}
\end{figure*}

Our passivity-based controller has several key modifications including the actuation at the toe joints and an acceleration reference design. One common feature of popular dynamic controllers \cite{gong2018feedback, HLIP_TRO_XIONG} for bipedal locomotion is the zero actuation of the ankle actuator (providing zero or a small damping torque to the ankle pitch and roll joints). The zero actuation allows the stance foot to quickly comply with the ground incline at the foot landing instant and makes the biped robot equivalent to a point-foot robot, consistent with the PIPM. These controllers rely on accurate foot placement to reach desired CoM velocity. A mismatch, however, usually exists between the PIPM and the full-order model. This mismatch can lead to non-trivial tracking errors, since the desired CoM velocity trajectory, analytically determined by the PIPM, is not an accurate description of how the full-order nonlinear system evolves. Therefore, a foot placement calculated by the PIPM, although accurately executed, would not necessarily drive the CoM to the desired velocity.
To this end, we adopt ankle control to compensate for this model mismatch. The desired acceleration of the torso is constructed using the feedback on the CoM state: $\boldsymbol{\ddot{q}}_{{\rm torso},xy} = g/h ((\boldsymbol{p}_{\rm com} - \boldsymbol{p}_{\rm foot}) + \boldsymbol{k}_{p,{\rm torso},xy} \boldsymbol{q}^e_{{\rm torso},xy} + \boldsymbol{k}_{d,{\rm torso},xy} \boldsymbol{\dot{q}}^e_{{\rm torso},xy})$, where $\boldsymbol{\ddot{q}}_{{\rm torso},xy}$ is the desired acceleration for the torso in the horizontal plane. $\boldsymbol{\ddot{q}}_{{\rm torso},xy}$ is a 2D subvector of its full vector, same for $\boldsymbol{q}^e_{{\rm torso},xy}$ and $\boldsymbol{\dot{q}}^e_{{\rm torso},xy}$. $\boldsymbol{k}_{p,{\rm torso},xy}$ and $\boldsymbol{k}_{d,{\rm torso},xy}$ are the task-level PD gains. The torso acceleration $\boldsymbol{\ddot{q}}_{{\rm torso},xy}$ results in additional control efforts for the stance leg including the ankle joints to trend the CoM toward the desired velocity.

\section{Implementation and Results}
\label{sec:results}
This result section evaluates the performance of (i) the high-level task planner by assessing its task completion, collision avoidance, and safe action execution; (ii) the middle-level motion planner by employing our designed keyframe decision maker to choose proper keyframe states and generating safe locomotion trajectories; and (iii) the low-level controllers for hardware implementation. 
The open-source code can be found here \href{https://github.com/GTLIDAR/safe-nav-locomotion.git}{\nolinkurl{https://github.com/GTLIDAR/safe-nav-locomotion.git}}. A video of the simulations is \href{https://youtu.be/oKGzubk2_9E}{\nolinkurl{https://youtu.be/oKGzubk2_9E}}. 

\subsection{LTL Task Planning Implementation}
The task planner is evaluated in an environment with multiple static and dynamic obstacles, and two rooms with different ground heights connected by a set of stairs as seen in Fig. \ref{fig:psp}. To generate the navigation planning abstraction, the environment is discretized into a $10\times5$ coarse grid, with a $2.7\times2.7$ m$^2$ cell size. $\mathcal{L}_{r,c}$ is the set of all accessible discrete cells, $\mathcal{H}_{r,c}$ is the set of cardinal directions, and $\mathcal{N}_a$ is a set of navigation actions in those cardinal directions (N, E, S, W).
Each coarse cell is further discretized into a finer $26\times26$ grid for local action planning. We model the possible actions as step length $d \in\{\textsf{small1},\textsf{small2},\textsf{medium1},\textsf{medium2},\textsf{large1},\textsf{large2}\}$, heading change $\Delta\theta \in \{\textsf{left}, \textsf{none}, \textsf{right}\}$), and step height $\Delta z_{\rm foot} \in \{z_{\rm down2},z_{\rm down1}, z_{\rm flat}, z_{\rm up1},z_{\rm up2}\}$. The possible heading changes $\Delta\theta \in \{-22.5^\circ, 0^\circ, 22.5^\circ\}$, are constrained by the minimum number of steps needed to make a $90^\circ$ turn, and the maximum allowable heading angle change that results in viable keyframe transitions as defined in Theorem~\ref{thm:steering1}. We choose $\Delta\theta = \pm 22.5^\circ$ so that a $90^\circ$ turn can be completed in four steps as shown in Fig. \ref{fig:non_determinstic}. Completing the turn in fewer steps is not feasible as it would overly constrain $v_{\rm apex}$, as can be seen in Fig. \ref{fig:steering_safety}(b). Due to the allowable heading change of $\pm 22.5^\circ$, $\mathcal{H}_{r,f}$ contains a discrete representation of the 16 possible headings the robot could have.


A set of specifications is designed to describe the allowable successor locations and actions in the transition system. Here, we only show a few specifications as examples: 
\begin{align}\nonumber
        &\square \big(( h_{r,f} = \mathcal{H}_{r,c}\wedge ((i_{\rm st} = \textsf{left} \wedge \Delta\theta = \textsf{right}) \\
        & \vee (i_{\rm st} = \textsf{right} \wedge \Delta\theta = \textsf{left})) \Rightarrow \bigcirc(d = \textsf{medium2})\big),\\\nonumber
        & \square \big(( h_{r,f} = \mathcal{H}_{r,c}\wedge ((i_{\rm st} = \textsf{left} \wedge \Delta\theta = \textsf{left}) \\
        & \vee (i_{\rm st} = \textsf{right} \wedge \Delta\theta = \textsf{right})) \Rightarrow \bigcirc(d = \textsf{small2})\big),
\end{align}
which govern the allowable step length during the first step of a turning process.

Both navigation and action planners are constructed by combining environment assumptions and system specifications generated by the successor functions described in Sections~\ref{subsec:task_planner_synth} and \ref{subsec:belief_multi} into a transition system and using the LTL synthesis tool SLUGS to generate a winning strategy. Synthesis occurs offline, and the winning strategy is efficiently encoded in a binary decision diagram (BDD) \cite{akers1978binary} which can be accessed online by interfacing the controller directly with SLUGS. At each turn of the game, the controller computes the new abstracted environment state and passes it to SLUGS which returns the corresponding system action.

\begin{figure}[t]
\centerline{\includegraphics[width=.48\textwidth]{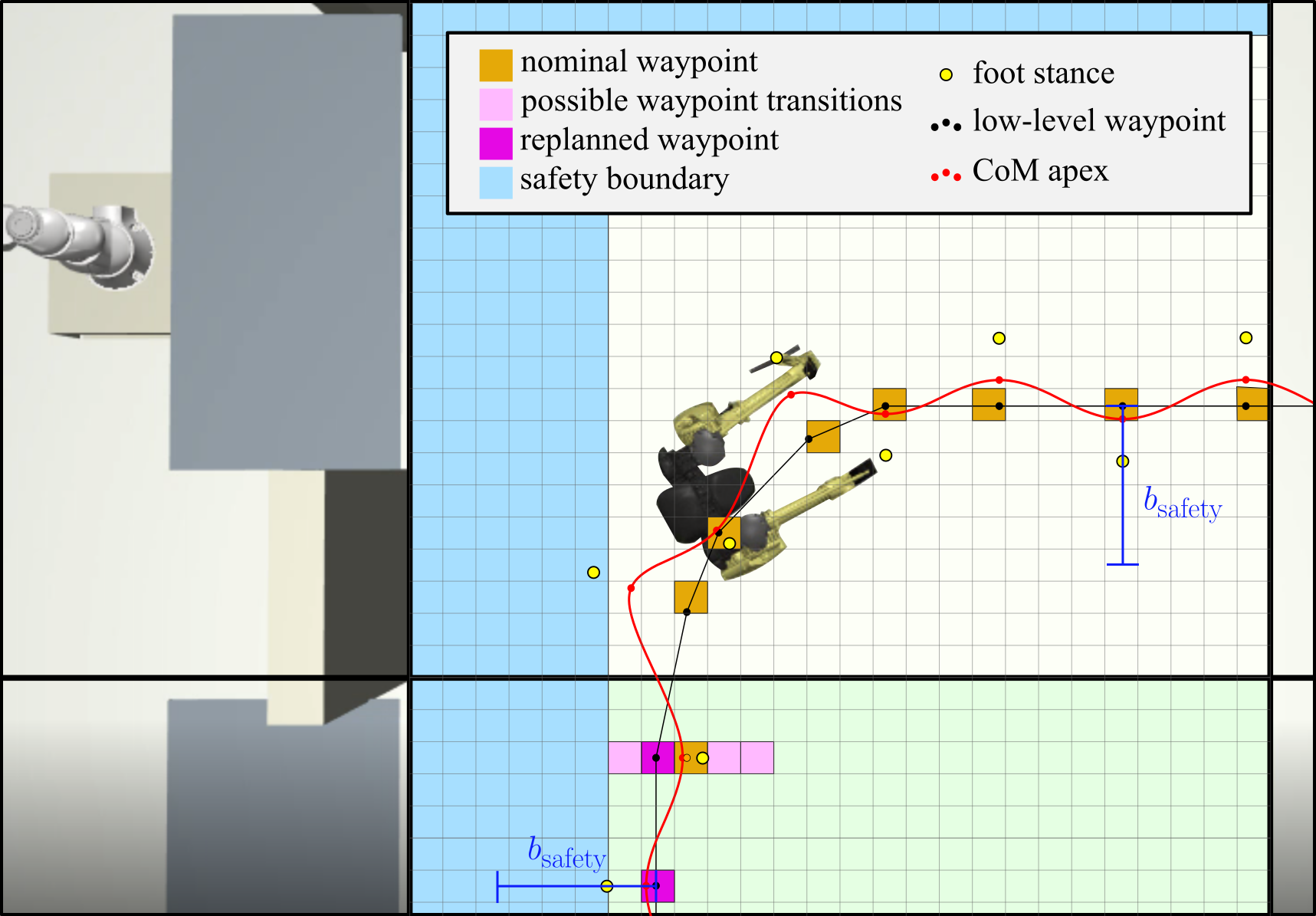}}
\caption{
Illustration of \textit{online} updating the high-level waypoint to maintain lateral tracking at the middle-level motion planner. The high-level waypoint is also required to keep a safe distance away from the adjacent coarse cell to avoid collisions with static or dynamic obstacles. In this run, we set the safety boundary to be $6$ fine cells as shown in light blue. 
}
\label{fig:top_down_turn}
\vspace{-0.05in}
\label{fig:result_turn}
\end{figure}

\subsection{Nominal Online Planning for A Pick and Place Task}
\label{subsec:results_nominal}
The middle-level motion planner is able to generate CoM trajectories of the ROM for a pick-and-place task infinitely often that includes traversing stairs, steering, stopping, and avoiding dynamic obstacles. The keyframe decision maker, detailed in Sec.~\ref{sec:Tracking}, selects the optimal next keyframe for waypoint tracking. The average execution time for one walking step in the middle level is $0.5$ ms. The action planner interfaces with the middle-level motion planner \textit{online} to pass the action set for the next keyframe. In the case when the keyframe decision maker cannot satisfy the lateral tracking of high-level waypoints in Proposition~\ref{thm:tracking_viability}, a new non-deterministic transition from the action planner is selected based on the modified lateral phase-space plan \textit{online}. The action planner receives the updated waypoint which allows the planner to choose the correct transition to the next game state. Our simulation shows that the robot successfully traverses uneven terrain to complete its navigation goals while steering away from dynamic obstacles when they appear in the robot's visible range. The robot's navigation trajectory is shown in Fig.~\ref{fig:psp}. The tracking results for multiple plans with different obstacle paths are detailed in Table.~\ref{tab:nominal_table} using PSP parameters given in Table.~\ref{tab:value_table}. Waypoint correction only occurs in the last step of a turning sequence due to the complexity of lateral tracking during steering scenarios. 
12 out of 260 steps\footnote{the step count refers to the number of high-level actions received by the middle-level motion planner, which includes stopping actions} result in alternative discrete state transitions in the lateral direction, all of which were seamlessly handled by the action planner as shown in Fig.~\ref{fig:result_turn}. This result shows that the integration of the high-level planner and the middle-level motion planner in an \textit{online} fashion allows for successful and safe TAMP.


\begin{table}[t]
    \centering
    \caption{Successful motion plan results for the pick and place task}
\begin{tabular}{ |c|c|c|  }
 \hline
\textbf {steps} & \textbf {turns} & \textbf {waypoint correction} \\
 \hline
 $200$ & $9$ & $4$ \\
 $260$ & $17$ & $12$ \\
 $500$ & $29$ & $22$ \\
 \hline

\end{tabular}

    \label{tab:nominal_table}
\end{table}
\begin{table}[t]
    \centering
    \caption{Nominal PSP parameters values}
\begin{tabular}{|c|c|c|c|}
 \hline
\textbf {parameter} & \textbf {value} & \textbf {parameter} & \textbf {value} \\
 \hline
 \hline
 $v_{{\rm apex, min}}$ & $0.20$ m/s & $v_{{\rm apex, max}}$ & $0.70$ m/s \\
 \hline
 $h_{\rm apex}$ & $0.985$ m & $\Delta z_{\rm foot}$ &$\{0, \pm 0.1, \pm 0.2\}$ m\\
 \hline
 \multirow{2}*{$\Delta y_{1, d}$} & $0.10$ m & \multirow{2}*{$\Delta y_{2, d}$} & \multirow{2}*{$0.14$ m} \\
 & ($0.0$ m for steering) & & \\
  \hline
 \multirow{2}*{$c_1$} & $1.0$ & \multirow{2}*{$c_2$} & $1.0$ \\
 & ($7$ for steering) & &  ($4$ for steering) \\
 \hline
  $c_3$ & $0$ & $c_4$ & $0$\\
 \hline
 $\Delta \theta$ & $\{0^\circ, \pm 22.5^\circ\}$ & $ b_{\rm safety}$ & $0.52$ m\\
 \hline

\multirow{3}*{$d$} & $\{0.21, 0.28, 0.31, $ & \multirow{3}*{$v_{\rm inc}$} & \multirow{3}*{$0.01$ m/s}\\
& $0.38, 0.42, 0.43,$ & & \\
& $0.47, 0.52\}$ m & & \\

 \hline
\end{tabular}

    \label{tab:value_table}
\end{table}
\vspace{-0.15in}




\subsection{Belief Space Planning}
The belief abstraction in the navigation planner is successful in tracking and bounding nonvisible obstacles as can be seen in Fig. \ref{fig:belief_trans}. The tracked belief enables the robot to navigate around static obstacles while guaranteeing that the dynamic obstacles are not in the immediate non-visible vicinity. Fig. \ref{fig:belief_results} depicts a snapshot of a simulation where the robot must navigate around such an obstacle to reach its goal states. The grid world environment is abstracted into 6 distinct belief regions resulting in 64 possible belief states. A successful strategy can be synthesized only when using a belief abstraction.
Without explicitly tracking possible non-visible obstacle locations, the task planner believes the obstacle could be in any non-visible cell when it is out of sight, including the adjacent visible cell in the next turn of the game. That means the planner can not guarantee collision avoidance and is not able to synthesize a strategy that would allow the robot to advance.
Fig. \ref{fig:Belief_sub2} depicts a potential collision that could occur in pink.
This comparison underlines the significance of the belief abstraction approach.

The belief abstraction provides additional information for deciding long-horizon navigation actions beyond guaranteeing immediate collision avoidance. In the simulation shown in Fig. \ref{fig:psp}, it is challenging to navigate around the vision occluding static obstacles at the lower-level (including the walls and a multi-stair platform). The synthesized strategy reacts to the additional information about the dynamic obstacle provided by belief tracking in three distinct ways. Based on the belief, the robot either (i) continues on the most direct route to the goal location; (ii) loops around to the right and positions itself to be able to go around either side of the static obstacle; or (iii) stops and waits until the dynamic obstacle disappears. The planner can choose any of these three strategies as long as all safety specifications are met. This non-deterministic mechanism offers the task planner flexibility in choosing safe navigation actions.


Generating global navigation task planners for two dynamic obstacles using a joint belief abstraction requires only 40\% of the synthesis time as that of independently tracking the belief state of each obstacle. Specifically, synthesizing a strategy for the scene in Fig. \ref{fig:psp} with two dynamic obstacles took 34 mins using joint belief tracking and 85 mins when individually tracking the belief of each obstacle.

\begin{figure}[t]
\centering
\begin{subfigure}{.5\columnwidth}
  \centering
  \includegraphics[width=.9\linewidth]{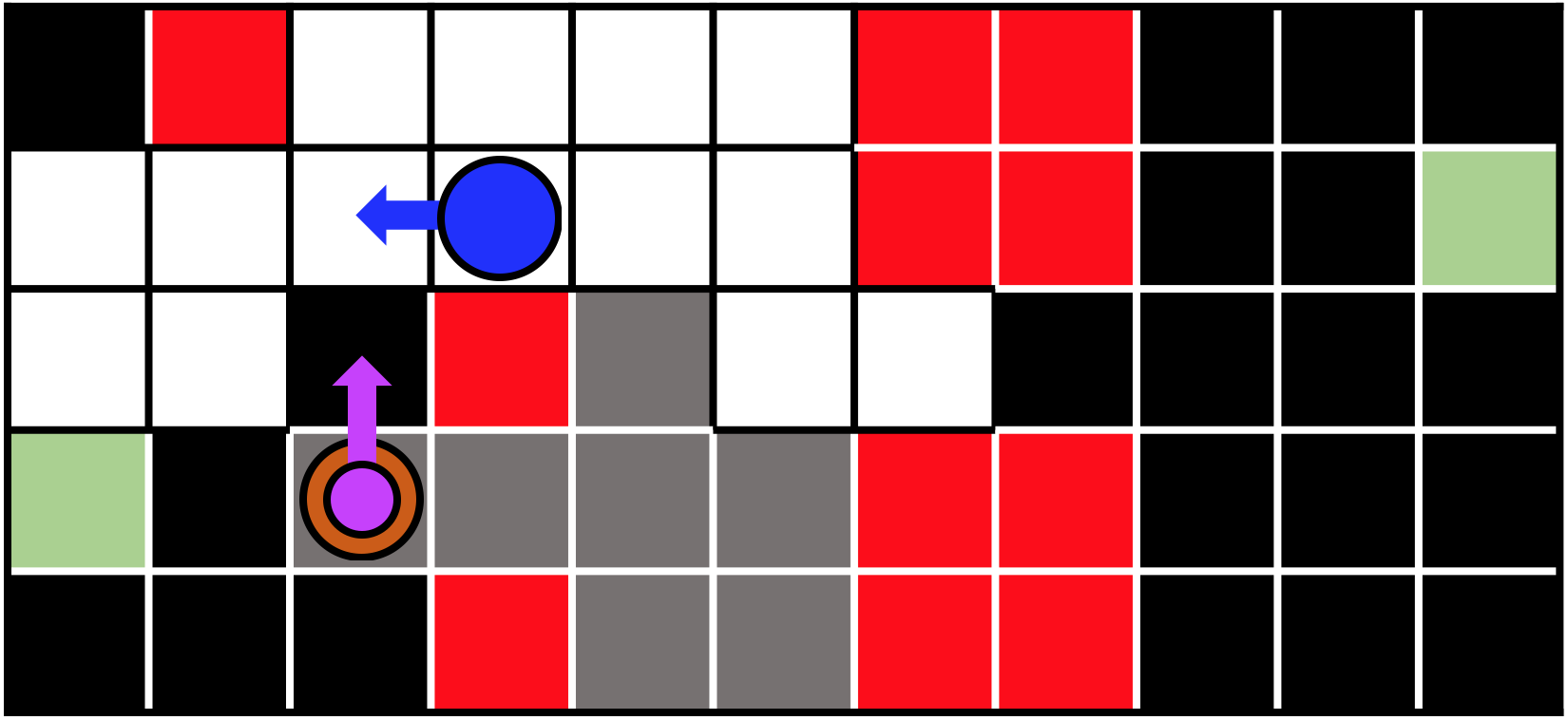}
  \caption{With explicit belief tracking}
  \label{fig:Belief_sub1}
\end{subfigure}%
\begin{subfigure}{.5\columnwidth}
  \centering
  \includegraphics[width=.9\linewidth]{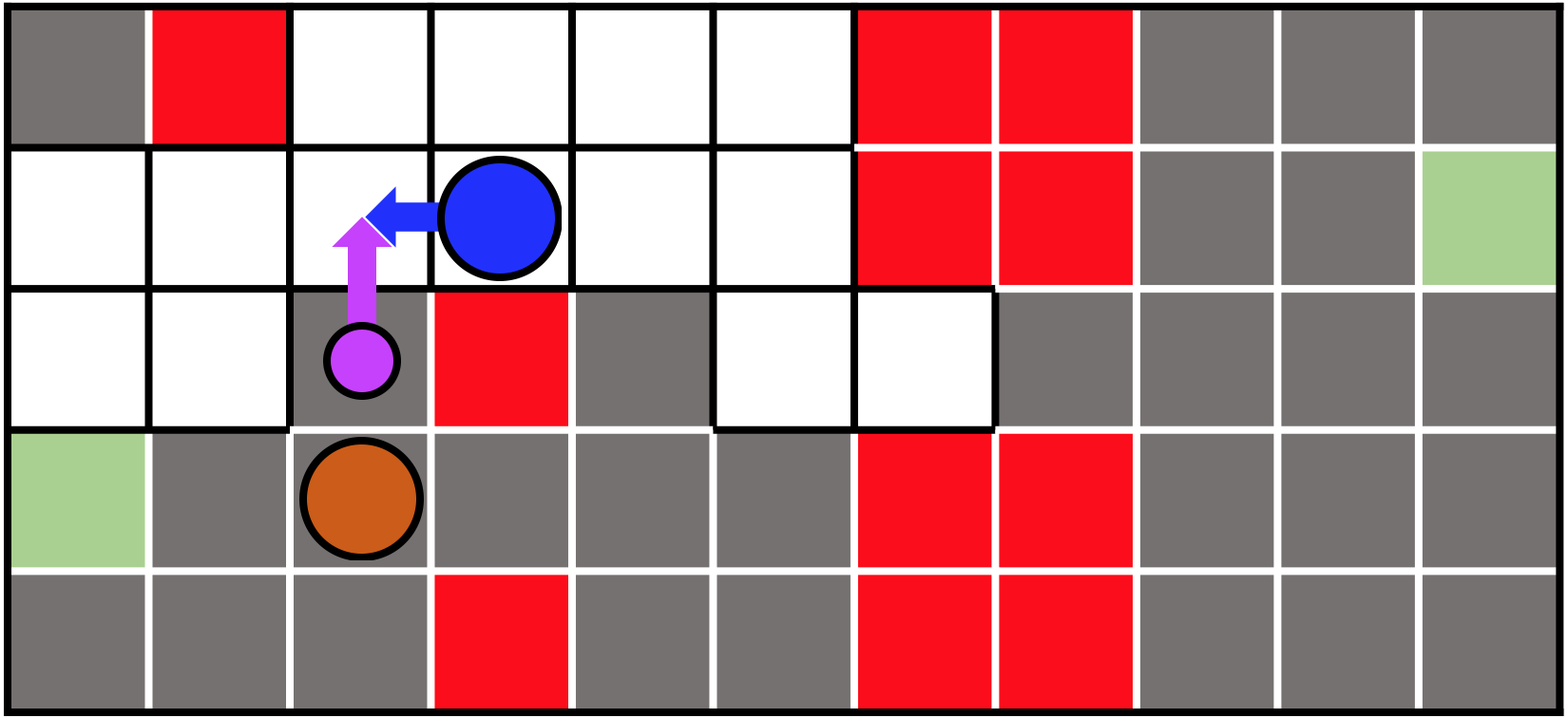}
  \caption{No explicit belief tracking}
  \label{fig:Belief_sub2}
\end{subfigure}
\caption{A snapshot of the coarse-level navigation grid during a simulation where the robot (blue circle) is going between the two goal states (green cells), while avoiding a static obstacle (red cells) and a dynamic obstacle (orange circle). White cells are visible while grey and black cells are non-visible. Gray cells represent the planner's belief of potential obstacle locations. The closest distance the obstacle could be to the robot, as believed by the planner, is depicted by the pink circle.
}

\label{fig:belief_results}
\vspace{-0.1in}
\end{figure}

\subsection{Hardware Experiment Setup and Results}
\label{sec:hardware_results}

\begin{figure}[t]
\centerline{\includegraphics[width=.48\textwidth]{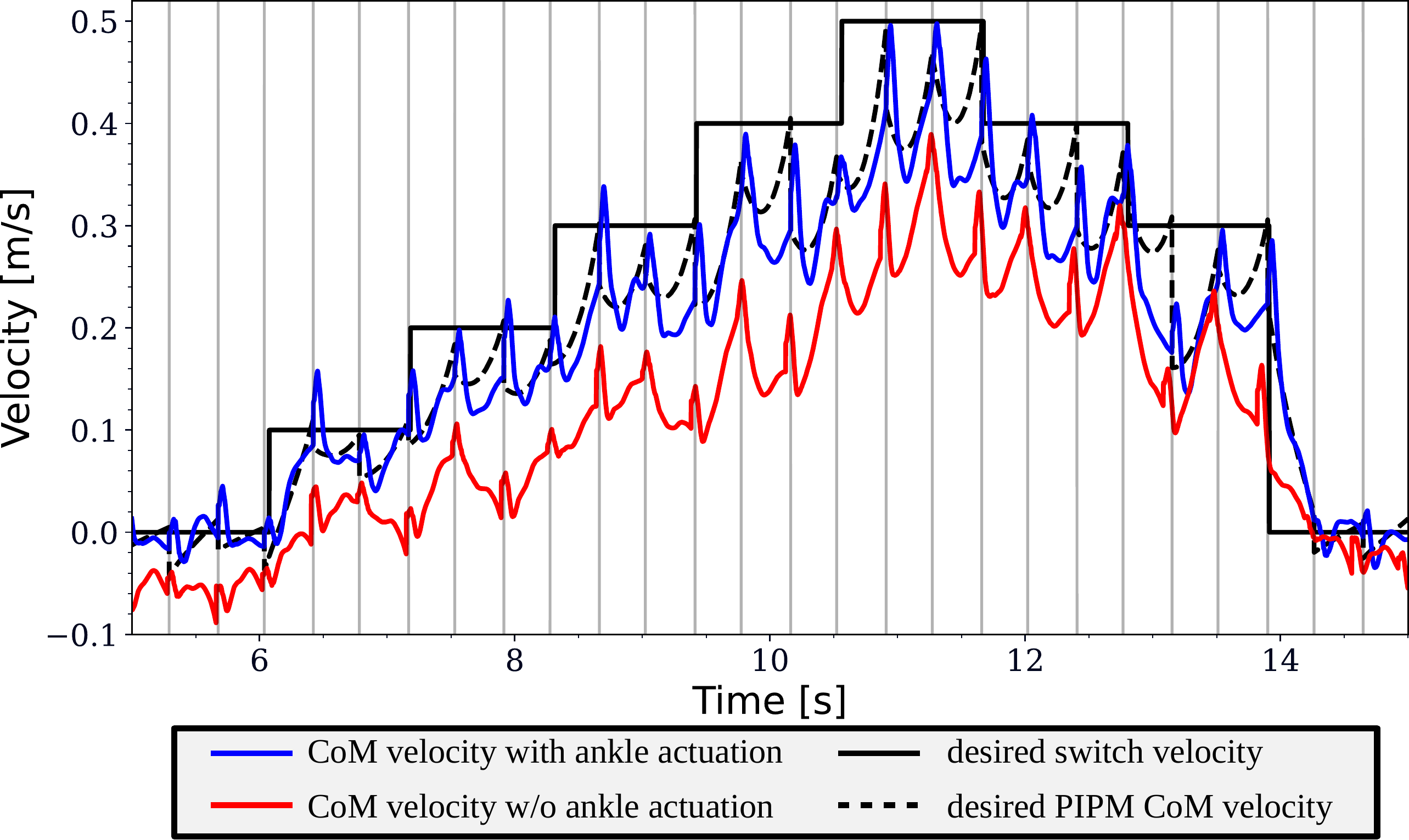}}
\caption{Sagittal velocity tracking performance of Digit hardware experiments. Given the same desired reference trajectory, the ankle-actuated velocity (shown in blue) achieves consistently small tracking errors, whereas the passive ankle velocity (shown in red) has a larger offset.}
\label{fig:vel_track}
\vspace{-0.05in}
\end{figure}

\begin{figure*}[t]
\centerline{\includegraphics[width=1\textwidth]{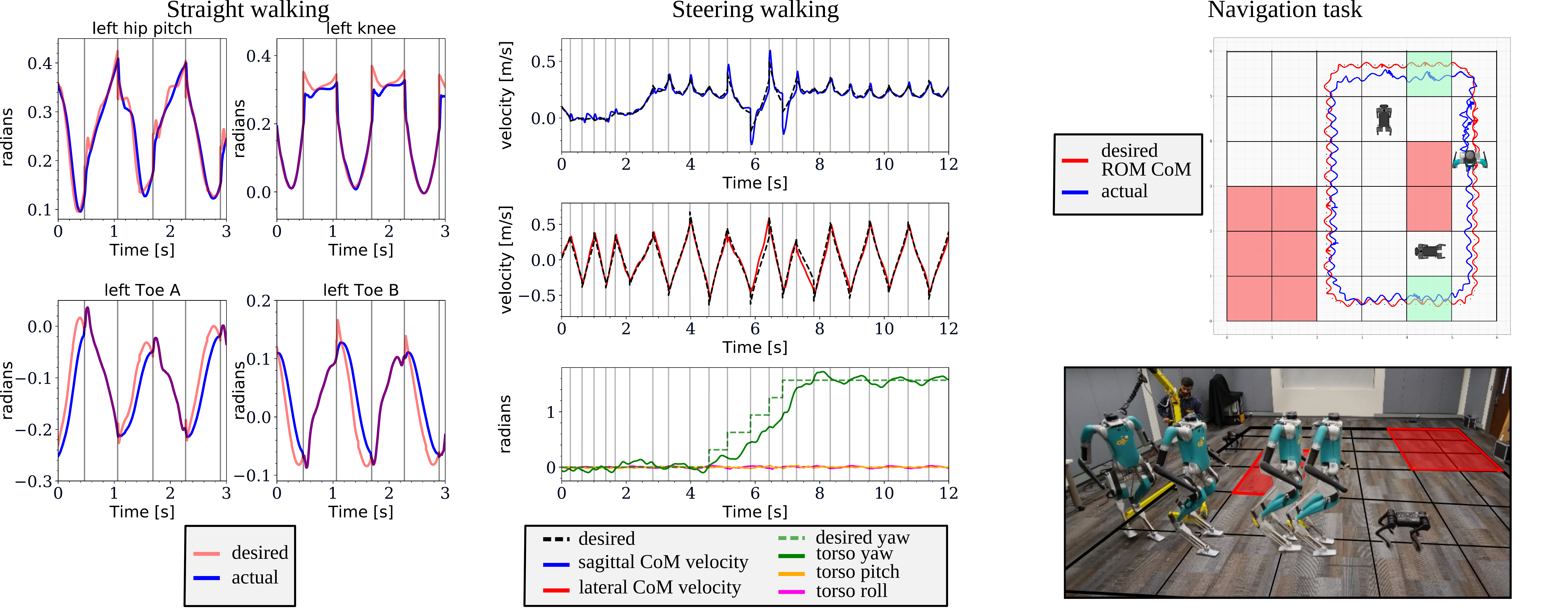}}
\caption{Hardware experiment results for Digit navigation tasks. For straight walking (left), we show multiple leg joint angles tracking performance. For steering walking (center), we show CoM sagittal and lateral velocity tracking performance in the top and the middle figures respectively, and torso Euler angles are shown in the bottom figure. The overall navigation task is shown in the right figure.}
\label{fig:hardware_results}
\vspace{-0.05in}
\end{figure*}

For Digit hardware experiments, we first test our velocity tracking performance in a straight walking setting. The achieved velocity tracking on Digit hardware is shown in Fig.~\ref{fig:vel_track}. Provided with the same desired reference trajectory, the ankle-actuated velocity (shown in blue) is much closer to the target and keeps the tracking error consistently small, whereas the passive ankle velocity (shown in red) has a larger offset. 
Fig.~\ref{fig:hardware_results} provides the tracking performance of the key joints on the left leg. The vertical black lines indicate the contact events. 
The tracking of the swing leg is important for achieving accurate foot placements. We achieve a satisfying performance---an error of less than 1 cm in the task space.
Fig.~\ref{fig:hardware_results} shows the torso's Euler angle during a $90^\circ$ turn while walking forward. The roll and pitch angles are kept flat at $0^\circ$, and the yaw angle tracks five consecutive $18^\circ$ step commands (dashed green line).



We demonstrate the efficacy of the entire framework through a navigation task in an open-world setting. The environment (see Fig.~\ref{fig:hardware_results}) is discretized into a $6\times6$ coarse grid, with a cell size scaled down to $1$ m$^2$ \footnote{The coarse cell size is scaled down for a compelling real-world implementation. Step length $d$ is scaled down accordingly and heading change is constrained to $\Delta \theta = \pm 18^\circ$ with five consecutive turning steps for $90^\circ$ turns.}, with flat ground, static obstacles, and two dynamic obstacles. The dynamic obstacles are Unitree A1 quadruped robots that perform their own navigation tasks sharing the same space. The Digit robot, as the protagonist, performs a navigation task---locomoting through the environment to two target coarse cells in sequence while avoiding all obstacles. The Digit robot can localize its own location with respect to the floor map known \textit{a priori}. 

The Digit robot runs the proposed framework online to conduct planning that generates the entire sequence of walking motions step by step\footnote{The parameters in Algorithm~\ref{alg:searching} are adjusted for hardware implementation, where $T_{d} = 0.45$ s, $W_{d} = 0.45$ m, $c_{1}= c_{2} = 4$, $c_{3} = 6$ and $c_{4} = 2$.}. This walking sequence needs to guarantee execution safety (i.e., obstacle collision avoidance and locomotion safety) and task completion (i.e., reaching the goal locations). The navigation result is shown in Fig.~\ref{fig:hardware_results}, where the Digit robot starts from the cell labeled with a Digit illustration and target coarse cells are shaded green. Along the way, one A1 gets in the way. Digit stops to avoid the collision, and then proceeds to finish the task after A1 moves away and the path is clear.

\section{Discussion and limitations}
\label{sec:discussion}
\textbf{Design and Computational Considerations of Bipedal Navigation: }Belief tracking expands the guaranteed safe navigation actions available to the navigation planner. Merging the belief of multiple dynamic obstacles into one abstract state captures less information than individual obstacle tracking by design. This reduces computational complexity while providing the same guarantees of capturing dynamic obstacle locations. One path to enhance the proposed framework in the future is to model small obstacles in the action planner so that an entire coarse cell containing such obstacles is still accessible to the robot in the navigation game.
Additionally, the library of turning sequences can be expanded to incorporate more aggressive navigation decisions while obeying the safety criteria proposed in Sec.~\ref{subsubsec:safetyprop} and to include fine-level obstacle avoidance maneuvers.


\textbf{Locomotion Safety Consideration in Real-world Deployment: }Our proposed safe PSP demonstrates the successful execution of high-level actions under nominal conditions for a large number of walking steps as detailed in Table.~\ref{tab:nominal_table}. While the framework still lacks a formal guarantee on successful lateral tracking of the high-level waypoints for an \textit{infinite} number of steps or under extremely large perturbations, our results show empirical guarantees afforded by the integration of the formal navigation and obstacle avoidance guarantees in the high-level task planner in Sec.~\ref{sec:task planner}, locomotion safety guarantees in Sec.~\ref{subsubsec:safetyprop}, and the online replanning algorithm for waypoint tracking in Sec.~\ref{sec:Tracking}. Such empirical guarantees are shown in Figs.~\ref{fig:psp}-\ref{fig:top_down_turn} and Table.~\ref{tab:nominal_table}.

The success rate of completing OWS safely under perturbation highly depends on various factors such as the state space granularity, robot actuation capability, environmental uncertainties, and the locomotion phase when the perturbation is applied. A comprehensive analysis of the success rate with respect to these factors is beyond the scope of this study. In the future, we plan to design more advanced safety criteria addressing adversarial pedestrians in the environment \cite{scianca2021behavior, zhi2021anticipatory, majd2021safe} and contact uncertainties from terrain \cite{drnach2021robust}. Moreover, the reachability analysis is based on the ROM dynamics, therefore a discrepancy will be induced when full-body dynamics is taken into account.

Moreover, in practice, locomotion safety is difficult to be guaranteed at the low level for the full-order-dynamics-based controller. Although existing works on CBF \cite{CBF_CLF_Ames} provided guarantees for full-order bipedal locomotion models, these guarantees can be easily violated due to various hardware implementation issues such as unmodeled actuator dynamics, communication delay, and imperfect state estimation. Given this fact, we have designed a full-order-dynamics-based controller implementation to maximize the success rate of task completion and safety from a practical perspective. Practically critical modifications to the controller, such as time-varying step length and toe actuation, have been incorporated to enable safety, although a theoretical guarantee can not be provided.

\textbf{Hardware Implementation Issues Induced by Robot Model Discrepancy: }Implementing such a planning framework on hardware gives rise to two main challenges, one relating to the high-level planner and one to the middle-level phase-space planner.

First, the abstraction of the robot state, while giving rise to task completion guarantees, also limits the actions that the robot can take to a specific granularity. Although we guarantee that our chosen granularity produces a safe motion plan for the reduce-order model, such safety guarantees might be compromised at the hardware level.

The middle-level phase-space planner is 
computationally efficient and allows for non-periodic motion plans. Discrepancy issues, however, will always arise when we use the plan of a ROM to control a full-order system. 
For example, the phase-space plan does not contain reference trajectories for the body orientation. In turning cases, a low-level full-body motion generator has to bridge the gap and design trajectories for the body orientation. We used task-space heuristics and inverse kinematics to generate the full-body trajectory for both turning and non-turning walking (see Sec.~\ref{sec:passivity}).
Another problem caused by the model mismatch is CoM velocity tracking. The desired CoM velocity trajectory, which is analytically determined by the PIPM, is not an accurate description of how the full-order nonlinear system evolves. Therefore, a foot placement calculated by the PIPM, although accurately executed, wouldn't necessarily drive the CoM to the desired velocity. This requires either a foot placement adjustment based on the full-order dynamics or additional regulation. We chose the second option and actuated the ankle torque in the low-level controller to provide better CoM velocity tracking performance, as shown in Fig.~\ref{fig:vel_track}.

\section{Conclusion}
\label{sec:conclusion}

Long-horizon and formally-guaranteed safe TAMP in complex environments with dynamic obstacles has long been a challenging problem, specifically for underactuated bipedal systems. On the other hand, symbolic planners are powerful in providing formal guarantees on safety and task completion in complex environments. For this reason, integrating high-level formal methods and low-level safe motion planning ought to be explored by the locomotion community to attain formally safe TAMP for real-world applications, to name a few, first-responders during search and rescue scenes, monitoring a controlled environment in civil or mechanical infrastructures, and planting crop products in agricultural environments. The way we address this problem is through multi-level safety in a hierarchically integrated planning framework.

Our proposed TAMP framework seamlessly integrates low-level locomotion safety specifications into a formal high-level LTL synthesis, to guarantee the safe execution of high-level commands. The middle-level motion planner generates non-period motion plans that accurately execute safe high-level actions. Our high-level planner employs a belief abstraction to address the partial observability of a large environment and guarantees safe navigation. We also investigate robustness against external perturbation through safe sequential composition of keyframe states to achieve robust locomotion transitions. By employing an online foot placement controller and a full-body passivity-based controller, the overall task and motion planning framework is also validated on a 28-DoFs Digit bipedal robot.

\appendices
\section{Analytical Solution for PIPM Dynamics}
\label{appendix1}
When the CoM motion is constrained within a piece-wise linear surface parameterized by $h = a(x -x_{{\rm foot}})+h_{\rm apex}$, the ROM becomes linear and an analytical solution exists:
\begin{align}
    \label{eqn:appendix1}
    &x(t) = Ae^{\omega t} + Be^{-\omega t} + x_{\rm foot} \\
    \label{eqn:appendix2}
    &\dot{x}(t) = \omega (Ae^{\omega t} - Be^{-\omega t})
\end{align}
where $\omega = \sqrt{\frac{g}{h_{\rm apex}}},
    A = \frac{1}{2}((x_{0}-x_{\rm foot})+\frac{\dot{x}_{0}}{\omega}),
    B = \frac{1}{2}((x_{0}-x_{\rm foot})-\frac{\dot{x}_{0}}{\omega})$.
manipulate Eq. (\ref{eqn:appendix1})-(\ref{eqn:appendix2}) gives
\begin{equation}
    x+\frac{\dot{x}}{\omega}-x_{\rm foot}=2Ae^{\omega t}
\end{equation}
which renders
\begin{equation}\label{eqn:appendix3}
    t=\frac{1}{\omega}\log(\frac{x+\frac{\dot{x}}{\omega}-x_{\rm foot}}{2A})
\end{equation}

To find the dynamics, $\dot{x} = f(x)$, which will lead to the switching state solution, remove the $t$ term by plugging Eq. (\ref{eqn:appendix3}) into Eq. (\ref{eqn:appendix1}).
\begin{equation}
    \frac{1}{2}(x-\frac{\dot{x}}{\omega}-x_{\rm foot})= \frac{2AB}{x+\frac{\dot{x}}{\omega}-x_{\rm foot}}
\end{equation}
\begin{equation}
    (x-x_{\rm foot})^{2}-(\frac{\dot{x}}{\omega})^{2}=4AB
\end{equation}
which yields
\begin{equation}
    \dot{x}=\pm \sqrt{\omega^{2}((x-x_{\rm foot})^{2}-(x_{0}-x_{\rm foot})^{2})+\dot{x}_{0}^{2}}
\end{equation}

If the apex height is constant, then $\omega$ is constant. According to the constraint that sagittal velocity should be continuous, the sagittal switching position is obtained by 
\begin{equation}
    x_{\rm switch} = \frac{1}{2}(\frac{C}{x_{{\rm foot}, n}-x_{{\rm foot}, c}}+(x_{{\rm foot}, c}+x_{{\rm foot}, n}))
\end{equation}
where
\begin{equation}
\begin{split}
    C = & (x_{{\rm apex},c}-x_{{\rm foot},c})^{2}-(x_{{\rm apex},n}-x_{{\rm foot}, n})^{2} \\ 
    & + \frac{v_{{\rm apex},n}^{2}-v_{{\rm apex},c}^{2}}{\omega^{2}}
\end{split}
\end{equation}

\section{Proof of Theorem~\ref{thm:straight}}
\label{appendiX:proof_prop_straight}
\begin{proof}
First, the sagittal switching position can be obtained from the analytical solution in Appendix~\ref{appendix1}:
\begin{equation}
\label{eq:x_switch}
    x_{\rm switch} = \frac{1}{2}(\frac{C}{x_{{\rm foot}, n}-x_{{\rm foot}, c}}+(x_{{\rm foot}, c}+x_{{\rm foot}, n}))
\end{equation}
where $C = (x_{{\rm apex},c}-x_{{\rm foot},c})^{2}-(x_{{\rm apex},n}-x_{{\rm foot}, n})^{2} +(\dot{x}_{{\rm apex},n}^{2}-\dot{x}_{{\rm apex},c}^{2})/\omega^{2}$. This walking step switching position is required to stay between the two consecutive CoM apex positions, i.e., 
\begin{equation}
    x_{{\rm apex},c} \le x_{\rm switch} \le x_{{\rm apex},n}
\end{equation}
which introduces the sagittal apex velocity constraints for two consecutive keyframes as follows.
\begin{equation}
\begin{split}
    \omega^{2}(x_{{\rm apex},n}-x_{{\rm apex},c})&(x_{{\rm apex},c}+x_{{\rm apex},n}-2x_{{\rm foot},n}) \\
    \le \; v_{{\rm apex},n}^{2} & -v_{{\rm apex},c}^{2} \le  \\
    \omega^{2}(x_{{\rm apex},n}-x_{{\rm apex},c})&(x_{{\rm apex},c}+x_{{\rm apex},n}-2x_{{\rm foot},c})
\end{split}
\end{equation}

Given this bounded difference between two consecutive CoM apex velocity squares, the corresponding safe criterion for straight walking can be expressed as Eq.~(\ref{eq:straight}). 
\end{proof}

\section{Proof of Theorem~\ref{thm:steering1}}
\label{appendiX:proof_prop_steering}
\begin{proof}
First, for the sagittal phase-space, the sagittal velocity  is required to be above the asymptote:
\begin{equation}
    \dot{x}_{{\rm apex},c} \geq \omega \cdot x_{{\rm foot},c}
    \label{eq:steering_xdot}
\end{equation}
Initiating a heading angle change introduces a new local sagittal coordinates as seen in Fig.~\ref{fig:steering_safety}. Therefore Eq.~(\ref{eq:steering_xdot}) becomes
\begin{equation}
    v_{{\rm apex},c}\cdot \cos{\Delta \theta} \geq \omega \cdot \Delta y_{2,c} \cdot \sin{\Delta \theta}
    \label{eq:steering_sagittal}
\end{equation}
As for the lateral phase-space, the lateral velocity is required to be below the asymptote in the new coordinate as follows
\begin{equation}
    v_{{\rm apex},c}\cdot \sin{\Delta \theta} \leq \omega \cdot \Delta y_{2,c} \cdot \cos{\Delta \theta}
\label{eq:steering_lateral}
\end{equation}
Combining Eqs.~(\ref{eq:steering_sagittal}-\ref{eq:steering_lateral}) results in the steering safety criterion in Eq.~(\ref{eq:steering_prop}).
\end{proof}

\section*{Acknowledgment}
The authors would like to express our gratefulness to Suda Bharadwaj and Ufuk Topcu for their discussions on belief abstraction, Yinan Li and Jun Liu for their assistance in the reachability controller implementation, and Jialin Li for his early-stage help in setting up our Cassie robot visualization. Special thanks to Victor Paredes, Guillermo Castillo, and Ayonga Hereid for their support on Digit controller implementation.

\bibliographystyle{IEEEtran}
\bibliography{lidar.bib}

%

\end{document}